\documentclass[11pt,a4paper,fleqn]{article}
\pdfoutput=1

\usepackage{amsthm}
\usepackage{amssymb}
\usepackage{amstext}
\usepackage{amsmath}
\usepackage{url}
\usepackage{hyperref}
\usepackage[margin = 40pt,font=small,labelfont=bf]{caption}
\usepackage{arydshln}
\usepackage{pdfsync}
\usepackage[mac]{inputenc}
\usepackage[T1]{fontenc}
\usepackage{color}

\usepackage{algorithm}
\usepackage{algpseudocode}

\algnewcommand\algorithmicinput{\textbf{Input:}}
\algnewcommand\INPUT{\item[\algorithmicinput]}

\algnewcommand\algorithmicoutput{\textbf{Output:}}
\algnewcommand\OUTPUT{\item[\algorithmicoutput]}

\algnewcommand\algorithmicbegin{\textbf{Algorithm:}}
\algnewcommand\BEGIN{\item[\algorithmicbegin]}

\usepackage{enumerate}
\usepackage{bbm}
\usepackage{ae,aecompl}
\usepackage{pdfpages}
\usepackage{epstopdf}
\usepackage{graphicx}
\graphicspath{{./images/}}
\usepackage{epsfig}
\usepackage{caption}
\usepackage{subcaption}

\numberwithin{equation}{section}

\newcommand{\thefont}[2]{\fontsize{#1}{#2}\fontshape{n}\selectfont}
\newcommand{\1}{\rlap{\thefont{10pt}{12pt}1}\kern.16em\rlap{\thefont{11pt}{13.2pt}1}\kern.4em}

\def\argmin{\mathop{\rm arg \; min}\limits}%

\theoremstyle{plain}
\newtheorem{proposition}{Proposition}[section]
\newtheorem{corollary}{Corollary}[section]
\newtheorem{theorem}{Theorem}[section]
\newtheorem{lemma}{Lemma}[section]

\theoremstyle{definition}
\newtheorem{definition}{Definition}[section]

\theoremstyle{remark}
\newtheorem{remark}{Remark}[section]

\newcommand{\E}{{\mathbb E}}
\newcommand{\R}{{\mathbb R}}
\newcommand{\C}{{\mathbb C}}

\newcommand{\N}{{\mathbb N}}

\renewcommand{\P}{{\mathbb P}}
\newcommand{\Id}{{\mathrm{Id}}}

\newcommand{\HH}{\ensuremath{\mathcal H}}

\def\argmin{\mathop{\rm arg \; min}\limits}%
\newcommand{\supp}{\mathop{\mathrm{supp}}}

\usepackage{tikz}
\usepackage{pgfplots}
\pgfplotsset{compat=newest}
\usetikzlibrary{plotmarks}
\usetikzlibrary{shadows}
\usepackage{grffile}
\usepackage{subcaption}
\usetikzlibrary{fit}
\usetikzlibrary{patterns,arrows,decorations.pathreplacing}

\title{Estimation of linear operators from scattered impulse responses}

\author{ J\'{e}r\'{e}mie Bigot\footnote{Institut de Math\'ematiques de Bordeaux et CNRS, IMB-UMR5251, Universit\'e de Bordeaux, {\tt jeremie.bigot@math.u-bordeaux.fr}}  \hspace{1cm}  Paul Escande\footnote{D\'epartement d'Ing\'{e}nierie des Syst\`{e}mes Complexes (DISC), Institut Sup\'{e}rieur de l'A\'{e}ronautique et de l'Espace (ISAE), Toulouse, France, {\tt paul.escande@gmail.com}}  \hspace{1cm} Pierre Weiss\footnote{Institut des Technologies Avanc\'{e}es en Sciences du Vivant, ITAV-USR3505 and Institut de Math\'{e}matiques de Toulouse, IMT-UMR5219, CNRS and Universit\'{e} de Toulouse, Toulouse, France, {\tt pierre.armand.weiss@gmail.com}}  }

\date{\today}

\begin{document}

\maketitle
\thispagestyle{empty}

\begin{abstract}
We provide a new estimator of integral operators with smooth kernels, obtained from a set of scattered and noisy impulse responses.  The proposed approach relies on the formalism of smoothing in reproducing kernel Hilbert spaces and on the choice of an appropriate regularization term that takes the smoothness of the operator into account. It is numerically tractable in very large dimensions.  We study the estimator's robustness to noise and analyze its approximation properties with respect to the size and the geometry of the dataset. In addition, we show minimax optimality of the proposed estimator.
\end{abstract}

\noindent \emph{Keywords:} Integral operator, scattered approximation, estimator, convergence rate, numerical complexity, radial basis functions, Reproducing Kernel Hilbert Spaces, minimax. \\

\noindent\emph{AMS classifications:} 47A58, 41A15, 41A25, 68W25, 62H12, 65T60, 94A20.

\section*{Acknowledgments} 
The authors wish to acknowledge the excellent reviewing work for this paper. Important technical inconsistencies were pointed out in the first version of the paper, which helped improving the manuscript substantially. The great care taken  here has become very rare and we are truly indebted to the reviewers.
The authors are also grateful to Bruno Torr\'esani and R\'emi Gribonval for their interesting comments on a preliminary version of this paper.
The PhD degree of Paul Escande has been supported by the MODIM project funded by the PRES of Toulouse University and the Midi-Pyr\'en\'ees R\'egion. This work was partially supported by the OPTIMUS project from RITC.

\section{Introduction} \label{sec:intro}

Let $H:L^2(\R^d) \to L^2(\R^d)$ denote a linear integral operator defined for all $u\in L^2(\R^d)$ and $x \in \R^d$ by:
\begin{equation} \label{eq:integral_op}
Hu(x) = \int_{\R^d} K(x,y) u(y) dy,
\end{equation}
where $K: \R^d \times \R^d \to \R$, is the operator kernel.
Given a set of functions $(u_i)_{1\leq i \leq n}$, the problem of operator identification consists of recovering $H$ from the knowledge of $f_i=H u_i+\epsilon_i$, where $\epsilon_i$ is an unknown perturbation. 

This problem arises in many fields of science and engineering such as mobile communication \cite{kailath1959sampling}, imaging \cite{gentile2013interpolating} and geophysics \cite{belanger2015compressed}. 
Many different reconstruction approaches have been developed, depending on the operator's regularity and the set of test functions $(u_i)$.
Assuming that $H$ has a bandlimited Kohn-Nirenberg symbol and that its action on a Dirac comb is known, a few authors proposed extensions of Shannon's sampling theorem \cite{kailath1959sampling,kozek2005identification,pfander2013sampling,Grochenig2014}.
Another recent trend is to assume that $H$ can be decomposed as a linear combination of a small number of elementary operators.
When the operators are fixed, recovering $H$ amounts to solving a linear system.
The work \cite{chiu2012matrix} analyzes the conditioning of this linear system when $H$ is a matrix applied to a random Gaussian vector.
When the operator can be sparsely represented in a dictionary of elementary matrices, compressed sensing theories can be developed \cite{pfander2008identification}. 
Finally, in astrophysics, a few authors considered interpolating the coefficients of a few known impulse responses (also called Point Spread Functions, PSF) in a well chosen basis \cite{gentile2013interpolating,mboula2015super,chang2012atmospheric}. This strategy corresponds to assuming that $u_i=\delta_{y_i}$ and it is often used when the PSFs are compactly supported and have smooth variations. 
Notice that in this setting, each PSFs is known \emph{independently} of the others, contrarily to the work \cite{pfander2013sampling}.

This last approach is particularly effective in large scale imaging applications due to two useful facts.
First, representing the impulse responses in a small dimensional basis allows reducing the number of parameters to identify.
Second, there now exist efficient interpolation schemes based on radial basis functions. 
Despite its empirical success, this method still lacks of solid mathematical foundations and many practical questions remain open:
\begin{itemize}
 \item Under what hypotheses on the operator $H$ can this method be applied?
 \item What is the influence of the geometry of the set $(y_i)_{1\leq i\leq n}$?
 \item Is the reconstruction stable to the pertubations $(\epsilon_i)_{1\leq i\leq n}$? If not, how to make robust reconstructions, tractable in very large scale problems? 
 \item What theoretical guarantees can be provided in this challenging setting?
\end{itemize}

The objective of this work is to address the above mentioned questions. 
We design a robust algorithm applicable in large scale applications. 
It yields a finite dimensional operator estimator of $H$ allowing for fast matrix-vector products, which are essential for further processing.
The theoretical convergence rate of the estimator as the number of observations increases is studied thoroughly. 

The outline of this paper is as follows. 
We first specify the problem setting precisely in Section \ref{sec:problemsetting}.
We then describe the main outcomes of our study in Section \ref{sec:main_results}.
We provide a detailed explanation of the numerical algorithm in Section \ref{sec:numerical_algorithm}.
Finally, the proofs of the main results are given in Section \ref{sec:proofs}.

\section{Problem setting}
\label{sec:problemsetting}

Throughout the paper, $\Omega\subset \R^d$ will denote a bounded, open and connected set, with Lipschitz continuous boundary. 

The value of a function $f$ at $x$ is denoted $f(x)$, while the $i$-th value of a vector $v \in \R^{N}$ is denoted $v[i]$. 
The $(i,j)$-th element of a matrix $A$ is denoted $A[i,j]$. 
The Sobolev space $H^s(\Omega)$ is defined for $s$ in $\N$ by 
\begin{equation}\label{def:SobolevSpace}
H^s(\Omega)=\left\{u\in L^2(\Omega), \partial^{\alpha} u \in L^2(\Omega), \ \textrm{for all multi-index } \alpha \in \N^d \ s.t.  \ |\alpha|=\sum_{i=1}^d \alpha[i] \leq s\right\}.  
\end{equation}
The space $H^s(\Omega)$ can be endowed with a norm $\|u\|_{H^s(\Omega)} = \left( \sum_{|\alpha|\leq s} \|\partial^{\alpha} u\|_{L^2(\Omega)}^2 \right)^{1/2}$ and the semi-norm $|u|_{H^s(\Omega)} = \left( \sum_{|\alpha|= s} \|\partial^\alpha u\|_{L^2(\Omega)}^2\right)^{1/2}$.
In addition, we will use the Beppo-Levi semi-norm defined by
$|u|_{BL^s(\Omega)}^2 = \sum_{|\alpha|=s} \frac{s!}{\alpha_1!\alpha_2!\hdots\alpha_d!} \|\partial^\alpha u\|_{L^2(\Omega)}^2$
and the Beppo-Levi semi-inner product defined by
\begin{equation}
 \left\langle f,g\right\rangle_{BL^s(\Omega)} = \sum_{|\alpha|=s} \frac{s!}{\alpha_1!\alpha_2!\hdots\alpha_d!} \langle \partial^\alpha f,\partial^\alpha g\rangle_{L^2(\Omega)}.
\end{equation}

Let $a$ and $b$ denote two functions depending on a parameter $u$ living in a set $U$. The notation $a(u) \lesssim b(u)$ means that there exists a constant $c >0$ such that $a(u) \leq c b(u)$ for all $u\in U$, with $c$ independent of the parameters $u$.
The notation $a(u)\asymp b(u)$ means that $a$ and $b$ are equivalent, i.e. there exists $0<c\leq C$ such that $ca(u)\leq b(u)\leq Ca(u)$.

The Beppo-Levi and the Sobolev semi-norms are equivalent over the space $H^s(\R^d)$:
\begin{equation}
 |u|_{BL^s(\Omega)}^2 \asymp |u|_{H^s(\Omega)}^2.
\end{equation}

\subsection{The sampling model}

An integral operator can be represented in many different ways. 
A key representation in this paper is the Space Varying Impulse Response (SVIR) $S : \R^d \times \R^d \to \R$ defined for all $(x,y)\in \R^d\times \R^d$ by:
\begin{equation} \label{eq:defTVIR}
 S(x,y) = K(x+y,y).
\end{equation}
The impulse response or Point Spread Function (PSF) at location $y\in \R^d$ is defined by $S(\cdot,y)$.

The main purpose of this paper is the reconstruction of the SVIR  of an operator from the observation of a few  impulse responses $S(\cdot, y_i)$ at scattered (but known) locations $(y_i)_{1\leq i \leq n}$ in a set $\Omega$. 
In applications, the PSFs $S(\cdot, y_i)$ can only be observed through a projection onto an $N$ dimensional linear subspace $V_N$. We assume that the linear subspace $V_N$ reads
\begin{equation}\label{eq:defVN}
V_N=\mathrm{span}\left( \phi_k, 1\leq k\leq N\right), 
\end{equation}
where $(\phi_k)_{k\in \N}$ is an orthonormal basis of $L^2(\R^d)$.
In addition, the data is often corrupted by noise and we therefore observe a set of $N$ dimensional vectors $(F_i^\epsilon)_{1\leq i \leq n}$ defined for all $k\in \{1,\hdots,N\}$ by
\begin{equation}
 F_{i}^\epsilon[k] = \langle S(\cdot, y_i), \phi_k \rangle + \epsilon_{i}[k],  \; 1 \leq i \leq n,\label{eq:datamodel}
\end{equation}
where $\epsilon_{i}$ is a random vector with independent and identically distributed (iid) components with zero mean and finite variance $\sigma^2$. For \eqref{eq:datamodel} to be well defined, $S$ should be sufficiently smooth and we will provide precise regularity conditions in the next section.

Since impulse responses are observed on a bounded set $\Omega$, we can only expect reconstructing $S$ faithfully on $\R^d\times \Omega$ and not on the whole space $\R^d\times \R^d$. 
Hence the objective of this work is to define an estimator $\hat H$ with kernel $\hat K$ close to $H$ with respect to the Hilbert-Schmidt norm defined by:
\begin{equation}
 \|\hat H - H\|_{HS}^2 := \int_{\Omega} \int_{\R^d} |\hat K(x,y)- K(x,y)|^2 \,dx\,dy.
\end{equation}
Controlling the Hilbert-Schmidt norm allows controlling the action of $\hat H$ on functions compactly supported on $\Omega$. Indeed, for a function $u\in L^2(\R^d)$ with $\supp(u)\subseteq \Omega$, we get - using Cauchy-Schwarz inequality:
\begin{align*}
 \|\hat H u - H u \|_{L^2(\R^d)}^2 & = \int_{\R^d} \left(\int_{\R^d}  (\hat K(x,y)-K(x,y)) u(y) \,dy\right)^2 \,dx \\
 & \leq \int_{\R^d}  \|\hat K(x,\cdot)-K(x,\cdot)\|_{L^2(\Omega)}^2 \|u\|_{L^2(\R^d)}^2 \,dx \\
 &= \|\hat H - H\|_{HS}^2 \|u\|_{L^2(\R^d)}^2.
\end{align*}


\subsection{Space varying impulse response regularity}

The SVIR encodes the impulse response variations in the $y$ direction, instead of the $(x-y)$ direction for the kernel representation, see Figure \ref{fig:kernel_vs_tvir} for a 1D example. 
It is convenient since in many applications, the smoothness of $S$ in the $x$ and $y$ directions is driven by different physical phenomena. For instance in astrophysics, the regularity of $S(\cdot,y)$ depends on the optical system, while the regularity of $S(x,\cdot)$ may depend on exteriors factors such as atmospheric turbulence or weak gravitational lensing \cite{chang2012atmospheric}.
This property will be expressed through the specific regularity assumptions of $S$ defined hereafter.

First we will make use of the following functional space.
\begin{definition}\label{def:Er}
 The space $\mathcal{E}^r(\R^d)$ (also denoted $\mathcal{E}^r$) is defined, for all $r\in \R$ and $r>\frac{d}{2}$, as the set of functions $u\in L^2(\R^d)$ such that:
 \begin{equation}\label{eq:ineqsobolev}
  \|u\|_{\mathcal{E}^r(\R^d)}^2= \sum_{k\in \N} w[k] |\langle u,\phi_k \rangle|^2 < +\infty,
 \end{equation}
 where $w:\N\to \R_+^*$ is a weight sequence satisfying $w[k]\gtrsim (1+k^2)^{r/d}$.
\end{definition}

\begin{remark}
This definition is introduced in reference to the Sobolev spaces $H^{m}_{\Delta}(\R^d)$ of functions with $m$ derivatives in $L^2(\R^d)$ supported on a compact set $\Delta$. This space can be defined - alternatively to equation \eqref{def:SobolevSpace} - by:
\begin{equation}
 H^m_{\Delta}(\R^d)=  \left\{ u\in L^2(\R^d), \sum_{\lambda \in \Lambda} 2^{2m|\lambda|} |\langle u, \psi_\lambda \rangle|^2 <+\infty\right\},
\end{equation}
where $(\psi_\lambda)_{\lambda \in \Lambda}$ is a wavelet basis with at least $m+1$ vanishing moments (see e.g. \cite[Chapter 9]{Mallat-Book}) and $\lambda=(j,k)$ is a scale-space parameter. 
\end{remark}
\begin{remark}
Definition \ref{def:Er} encompasses many other spaces. 
For instance, it allows choosing a basis $(\phi_k)_{k\in \N}$ that is best adapted to the impulse responses at hand, by using principal component analysis, as was proposed in a few applied papers \cite{jee2007principal,berge2012point}. 
\end{remark}

The following definition gathers all the assumptions made on the operators. It will be used throughout the paper.

\begin{definition}\label{def:ballEr}
Let $A_{1}$ and $A_{2}$ be positive constants. 
Set $r>\frac{d}{2}$ and $s>\frac{d}{2}$.  
The ball $\mathcal{E}^{r,s}(A_{1},A_{2})$ is defined as the set of linear integral operators $H$ with SVIR $S$ belonging to $L^2(\R^d \times \R^d)$ with:
  \begin{align}
  &  \textrm{\textbf{Smooth variations}} & \quad & \int_{x \in \R^d} \|S(x,\cdot)\|_{H^s(\R^d)}^{2} dx \leq A_{1}\label{eq:varIR} \\ \nonumber \\
  & \textrm{\textbf{Impulse response regularity}} & \quad & \int_{y\in \R^d} \| S(\cdot,y)\|_{\mathcal{E}^r(\R^d)}^{2} dy \leq A_{2} \label{eq:regIR}  
  \end{align}
\end{definition}

Let us comment on these assumptions:
\begin{itemize}
 \item Equation \eqref{eq:regIR} means that $S(\cdot,y)$ belong to $\mathcal{E}^r(\R^d)$ for a.e. $y \in \R^d$.
 \item Similarly, assumption \eqref{eq:varIR} means that $S(x,\cdot)$ is in $H^s(\R^d)$ for a.e. $x \in \R^d$.
 The hypothesis $s>d/2$ ensures existence of a \emph{continuous} representant of $S(x,\cdot)$ for a.e. $x$, by Sobolev embedding theorems \cite[Thm.2, p.124]{stein2016singular}.
This regularity condition will allow the use of fine approximation results based on radial basis functions \cite{arcangeli2007extension}.
 \item The two regularity conditions are sufficient for the sampling procedure \eqref{eq:datamodel} to be well defined. Lemma \ref{lem:invert} indeed indicates that the functions $y \mapsto \langle S(\cdot,y), \phi_k \rangle$ are in $H^s(\R^d)$ for all $k \in \N$. By Sobolev embedding theorems \cite[Thm.2, p.124]{stein2016singular}, there exists a continuous representant of these functions and hence, we can give a meaning to $\langle S(\cdot,y), \phi_k \rangle$.
 \item In the particular case where $\mathcal{E}^r(\R^d) = H^r(\R^d)$ the space $\mathcal{E}^{r,s}$ is the mixed-Sobolev space $H^{r,s}(\R^d \times \R^d)$ \cite{lion1972non,opic1991estimates,zeiser2012wavelet}.
\end{itemize}

 \begin{figure}[htp]
	\centering
\begin{subfigure}[t]{0.3\columnwidth}\centering
  \vspace{0pt}
  \begin{tikzpicture} 
    \pgfmathsetmacro{\WDTH}{0.09}
    \node[anchor=south west,inner sep=0] (image) at (0,0) {\includegraphics[width=0.9\textwidth]{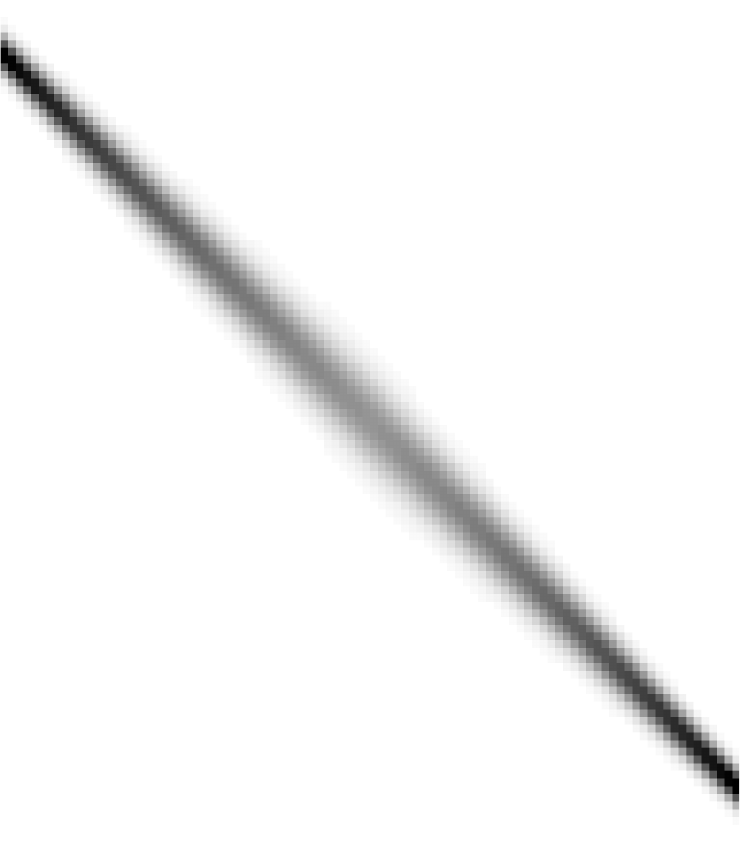}};
     \begin{scope}[x={(image.south east)},y={(image.north west)}]
	\draw [color=blue,line width=2pt,dashed] (0.135,0) -- (0.135,1) ;
	\draw [color=blue,line width=2pt,dashed] (0.865,0) -- (0.865,1) ;
	\draw [<->,color=blue,line width=1pt] (0.135,1) -- (0.865,1) ;
	\draw [color=blue] (0.5,1) node[above] {\Huge $\Omega$} ;
	\draw [->] (-0.1,1.2) -- (-0.1,0.5) ;
	\draw (0.3,1.2) node[above] {$y$} ;
	\draw (-0.1,0.9) node[left] {$x$} ;
	\draw [->] (-0.1,1.2) -- (0.6,1.2) ;
   \end{scope}
\end{tikzpicture}
\end{subfigure}%
\qquad\qquad\qquad
\begin{subfigure}[t]{0.3\columnwidth}\centering   \vspace{16pt}
  \begin{tikzpicture}
    \pgfmathsetmacro{\WDTH}{0.09}
    \node[anchor=south west,inner sep=0] (image) at (0,0) {\includegraphics[width=\textwidth]{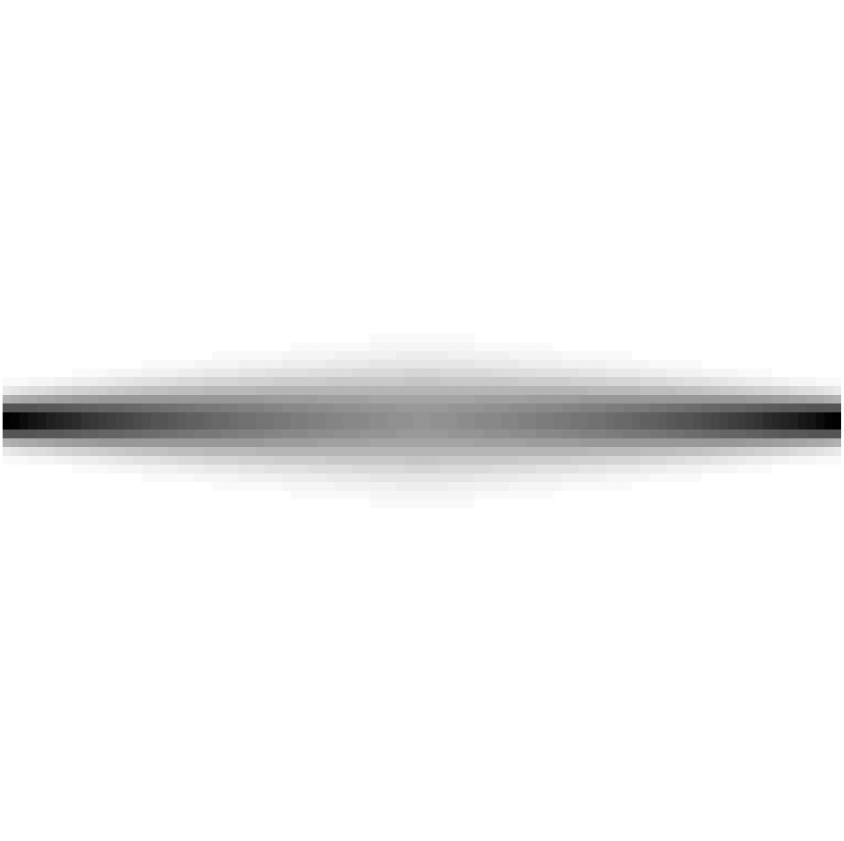}};
     \begin{scope}[x={(image.south east)},y={(image.north west)}]
	\draw [color=blue,line width=2pt,dashed] (0.135,0) -- (0.135,1) ;
	\draw [color=blue,line width=2pt,dashed] (0.865,0) -- (0.865,1) ;
	\draw [<->,color=blue,line width=1pt] (0.135,1) -- (0.865,1) ;
	\draw [color=blue] (0.5,1) node[above] {\Huge $\Omega$} ;
   \end{scope}
\end{tikzpicture}
\end{subfigure}
\caption{Illustration of the two representations of an integral operator considered in this paper. 
The kernel is defined by $K(x,y) = \frac{1}{\sqrt{2 \pi} \sigma(y)} \exp\left( -\frac{1}{2 \sigma(y)^2} |x-y|^{2} \right)$ for all $(x,y) \in \R^2$, where $\sigma(y) = 0.05 \left(1+ 2\min(y,1-y) \right)$ for $y \in \Omega = [0,1] $. Left: kernel representation (see equation \eqref{eq:integral_op}). Right: SVIR representation (see equation \eqref{eq:defTVIR}).
}
\label{fig:kernel_vs_tvir}
\end{figure}
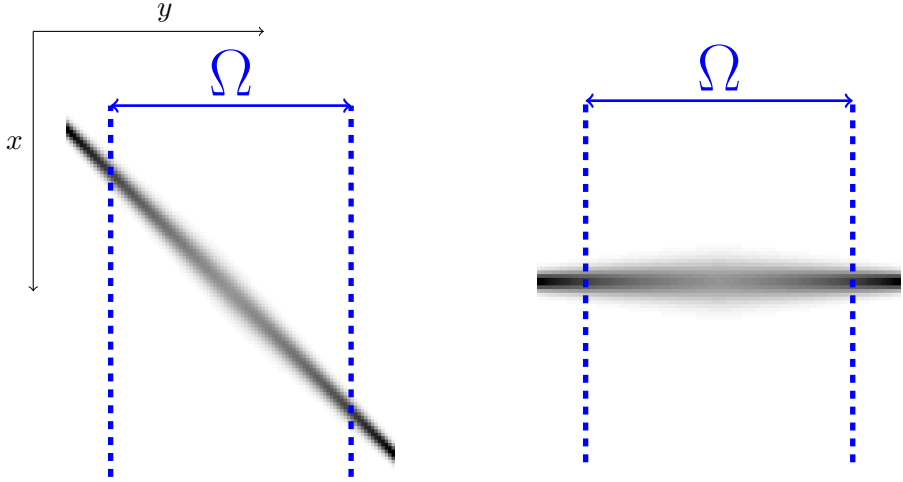

\section{Main results}\label{sec:main_results}


\subsection{Construction of an estimator}

Let $F: \R^d \to \R^{N}$ denote the vector-valued function representing the impulse responses coefficients (IRC) in basis $(\phi_k)_{k\in \N}$: 
\begin{equation} \label{eq:IRC}
F(y)[k] = \langle S(\cdot, y), \phi_k\rangle. 
\end{equation}

Based on the observation model \eqref{eq:datamodel}, a natural approach to estimate the SVIR, consists in constructing an estimate $\hat{F}: \R^d \to \R^N$ of $F$. 
The estimated SVIR is then defined as
\begin{equation}
 \hat{S}(x,y) = \sum_{k=1}^N \hat{F}(y)[k] \phi_{k}(x), \mbox{ for } (x,y) \in \R^d \times \R^d. \label{eq:hat_T}
\end{equation}

Definition \ref{def:ballEr} motivates the introduction of the following space.
\begin{definition}[Space $\mathcal{H}$ of IRC]\label{def:RKHS}
 The space $\mathcal{H}(\R^d)$ of admissible IRC is defined as the set of vector-valued functions $G:\R^d \to \R^N$ such that 
\begin{equation}
   \|G\|_{\mathcal{H}(\R^d)}^2 = \alpha \int_{y\in \R^d} \sum_{k=1}^N w[k] \left|G(y)[k]\right|^2\,dy + (1-\alpha) \sum_{k=1}^N \left| G(\cdot)[k] \right|_{BL^s(\R^d)}^2<+\infty,
\end{equation}
where $\alpha \in [0,1)$ allows to balance the smoothness in each direction. 
\end{definition}
The following result is straightforward (the proof is similar to that of Lemma \ref{lem:cv_rate_tvir_wavelet_row}).
\begin{lemma}
  Operators in $\mathcal{E}^{r,s}(A_1,A_2)$ have an IRC belonging to $\mathcal{H}(\R^d)$.
\end{lemma}
To construct an estimator of $F$, we propose to define $\hat{F}_{\mu}$ as the minimizer of the following optimization problem:
\begin{equation} 
\hat{F}_{\mu} = \argmin_{F \in \HH(\R^d)} \frac{1}{n} \sum_{i=1}^{n} \| F_i^\epsilon - F(y_i) \|^{2}_{\R^{N}} + \mu \| F \|^{2}_{\HH(\R^d)}, \label{eq:defhatF}
\end{equation}
where $\mu > 0$ is a regularization parameter. 

\begin{remark}
The proposed formulation can be interpreted with the formalism of regression and smoothing in vector-valued Reproducing Kernel Hilbert Spaces (RKHS) \cite{Micchelli04,Micchelli05}.
The space $\mathcal{H}(\R^d)$ can be shown to be a vector-valued Reproducing Kernel Hilbert Space (RKHS).
The formalism of vector-valued RKHS has been developed for the purpose of multi-task learning, and its application to operator estimation appears to be novel. 
\end{remark}

\subsection{Mixed-Sobolev space interpretation}

The problem formulation \eqref{eq:defhatF} might seem abstract at first sight. 
In this section we show that it encompasses the formalism of mixed-Sobolev spaces \cite{lion1972non,opic1991estimates,zeiser2012wavelet} and that the proposed methodology can be interpreted in terms of SVIR instead of IRC. 

\begin{lemma}\label{lem:mixedsobolevvinterpretation}
Suppose $H \in \mathcal{E}^{r,s}(A_1,A_2)$.
In the specific case where $(\phi_k)_{k \in \N}$ is a wavelet or a Fourier basis and $N=+\infty$, 
The cost function in Problem \eqref{eq:defhatF} is equivalent to 
\begin{equation} \label{eq:hatT_optim_problem}
\frac{1}{n} \sum_{i=1}^n \left\| F_i^\epsilon - \left(\langle S(\cdot, y_i), \phi_{k} \rangle\right)_{1\leq k \leq N} \right\|_{2}^2  
+ \mu \left( \alpha \int_{\R^d} \| S( \cdot, y) \|_{H^r(\R^d)}^2 dy 
+ (1-\alpha) \int_{\R^d} | S(x, \cdot) |_{BL^s(\R^d)}^2 dx \right). 
\end{equation}
\end{lemma}
\begin{proof}
 The proof is straightforward once showing the results in Lemma \ref{lem:invert}.
\end{proof}

This formulation is quite intuitive: the data fidelity term allows finding a TVIR that is close to the observed data, the first regularization term allows smoothing the additive noise on the acquired PSFs and the second one interpolates the missing data.

%
%

\subsection{Numerical complexity}

Thanks to the results in \cite{Micchelli05}, computing $\hat{F}_{\mu}$ amounts to solving a \emph{finite-dimensional} system of  linear equations. 
However, for an  arbitrary orthonormal basis $( \phi_k )_{k \in \N}$, and without any further assumptions on the kernel of the RKHS $\HH(\R^d)$, evaluating $\hat{F}_{\mu}$  leads to the resolution of a \emph{full} $nN \times nN$ linear system, which is untractable for large $N$ and $n$.  

With the specific choice of norm introduced in Definition \ref{def:RKHS}, the problem simplifies to the resolution of $N$ systems of equations of size $n \times n$. This step is investigated in details in Section \ref{sec:numerical_algorithm}. 
In this paragraph we gather the results describing the numerical complexity of the method.

\begin{proposition}\label{prop:complexity1}
  The solution of \eqref{eq:defhatF} can be computed in no more than $O(Nn^3)$ operations for any choice of basis $( \phi_k )_{k \in \N}$.
\end{proposition}
\begin{proof}
  See Section \ref{sec:numerical_algorithm}.
\end{proof}

In addition, if the weight function $w$ is piecewise constant, some $n\times n$ matrices are identical, allowing to compute an LU factorization once for all and using it to solve many systems. This yields the following result.

\begin{proposition}\label{prop:complexity2}
In the specific case where $(\phi_k)_{k\in \N}$ is a wavelet basis, it is natural to set function $w$ as a constant over each wavelet subband \cite[Thm. 9.4]{Mallat-Book}.  
Then, the solution of \eqref{eq:defhatF} can be computed in no more than $O\left(\frac{\log(N)}{d}n^3 + Nn^2\right)$ operations.
\end{proposition}
\begin{proof}
  See Section \ref{sec:numerical_algorithm}.
\end{proof}

Finally let us remark that for well chosen bases $(\phi_k)_{k\in \N}$ the impulse responses can be well approximated using a small number $N$ of atoms. Such instances of bases include Fourier bases, wavelet bases with appropriate properties and the basis formed with the principal components of the impulse responses. This makes the method tractable even in very large scale applications.

To conclude this paragraph, let us mention that the representation of an operator of type \eqref{eq:hat_T} can be used to evaluate matrix-vector products rapidly. We refer the interested reader to \cite{escande2016approximation} for more details.

\subsection{Convergence rates}

The convergence of the proposed estimator with respect to the number $n$ of observations is captured by the theorems of this section. We show that the approximation efficiency of our method depends on the geometry of the set of data locations, and - in particular - on the fill and separation distances defined below.

\begin{definition}[Fill distance] \label{def:fill_dist}
 The fill distance of $Y = \{ y_1, \ldots, y_n \} \subset \Omega$ is defined as:
 \begin{equation}
  h_{Y,\Omega} = \sup_{y \in \Omega} \min_{1 \leq j \leq n} \|y - y_j \|_2.
 \end{equation}
This is the distance for which any $y \in \Omega$ is at most at a distance $h_{Y,\Omega}$ of $Y$. It can also be interpreted as the radius of the largest ball with center in $\Omega$ that does not intersect $Y$.
\end{definition}

\begin{definition}[Separation distance] \label{def:speration_dist}
  The separation distance of $Y = \{ y_1, \ldots, y_n \} \subset \Omega$ is defined as:
 \begin{equation}
  q_{Y,\Omega} = \frac{1}{2} \min_{i \neq j} \|y_i - y_j \|_2.
 \end{equation}
 This quantity gives the maximal radius $r > 0$ such that all balls $\{y \in \R^d : \| y - y_j\|_2 < r \}$ are disjoint.
\end{definition}

The following condition plays a key role in our analysis \cite{narcowich1991norms,schaback1995error}.
\begin{definition}[Quasi-uniformity condition] \label{def:quasi-uniformity}
  A set of data locations $Y = \{ y_1, \ldots, y_n \} \subset \Omega$ is said to be quasi-uniform with respect to a constant $B > 0$ if
  \begin{equation}
    q_{Y,\Omega} \leq h_{Y,\Omega} \leq B q_{Y,\Omega}.
  \end{equation}
\end{definition}

\begin{remark}
Our main theorems will be stated under a quasi-uniformity condition of the sampling set. 
It is likely that this hypothesis can be refined using more stable reconstruction schemes as is commonly done in the reconstruction of bandlimited functions \cite{feichtinger1995efficient}.
\end{remark}

\begin{theorem} \label{thm:main_result}
Assume that $H \in \mathcal{E}^{r,s}(A_1, A_2)$ and that its SVIR $S$ is sampled using model \eqref{eq:datamodel} under the quasi-uniformity condition given in Definition \ref{def:quasi-uniformity}. Then the estimating operator $\hat H$ with SVIR $\hat S$ defined in equation \eqref{eq:hat_T} satisfies the following inequality
\begin{equation}\label{eq:main_result}
 \E \left( \| H - \hat H \|_{HS}^2 \right) \lesssim N^{-\frac{2r}{d}} + (N \sigma^2 n^{-1})^{\frac{2s}{2s+d}} (1-\alpha)^{-\frac{2s+2d}{2s+d}} ,
\end{equation}
for $\mu \propto (N \sigma^2 n^{-1})^{\frac{2s}{2s+d}} (1-\alpha)^{\frac{-d}{2s+d}}$.  This inequality holds uniformly on the ball $\mathcal{E}^{r,s}(A_{1},A_{2})$.
\end{theorem}

\begin{proof}
 See Section \ref{sec:proofs}.
\end{proof}

In applications where the user can choose the number of observations $N$ (e.g. if it is sufficiently large), the upper-bound \eqref{eq:main_result} can be optimized with respect to $N$.
\begin{corollary} \label{thm:main_result2}
Assume that $H \in \mathcal{E}^{r,s}(A_1, A_2)$ and that its SVIR $S$ is sampled using model \eqref{eq:datamodel} under the quasi-uniformity condition given in Definition \ref{def:quasi-uniformity}. Then the estimator $\hat H$ with SVIR $\hat S$ defined in equation \eqref{eq:hat_T} satisfies the following inequality
\begin{equation}\label{eq:main_result2}
 \E \left( \| H - \hat H \|_{HS}^2 \right) \lesssim (\sigma^2 n^{-1} (1-\alpha)^{-\left(d/s+1\right)})^{\frac{2q}{2q+d}},
\end{equation}
with the relation $1/q = 1/r + 1/s$, for $\mu \propto (\sigma^2 n^{-1})^{\frac{2q}{2q+d}} (1-\alpha)^{\frac{-d}{2s+d}}$ and $N \propto (\sigma^{2} n^{-1})^{-\frac{dq}{r(2q+d)}} (1-\alpha)^{\frac{(d^2+sd)q}{rs(2q+d)}}$.
This inequality holds uniformly on the ball $\mathcal{E}^{r,s}(A_{1},A_{2})$.
\end{corollary}
\begin{proof}
 See Section \ref{sec:proofs}.
\end{proof}

Corollary \ref{thm:main_result2} gives some insights on the estimator behavior. In particular:
\begin{itemize}
 \item It provides an explicit way of choosing the value of the regularization parameter $\mu$: it should decrease as the number of observations increases.
 \item If the number of observations $n$ is small, it is unnecessary to project the impulse responses on a high dimensional basis (i.e. $N$ large). The basic reason is that not enough information has been collected to reconstruct the fine details of the kernel.
 \item The optimal value of $\alpha$ in the corollary is $\alpha=0$, suggesting that the best option is to not use the additional regularizer $\alpha \int_{y\in \R^d} \sum_{k=1}^N w[k] \left|G(y)[k]\right|^2\,dy$. 
 This phenomenon is due to a rough upper-bound in the proof. Unfortunately, we did not manage to obtain finer estimates of some eigenvalues in the proof.  From a practical perspective, we observed a good behavior of this additional term in our numerical experiments and therefore decided to present the theory including this regularizer.
\end{itemize}


Finally, to conclude this section on convergence rates, it is shown that, under mild assumptions on the basis $(\phi_{k})_{k \geq 1}$, the rate of convergence $(\sigma^2 n^{-1})^{\frac{2q}{2q+d}}$ in inequality \eqref{eq:main_result2}  is optimal  in the case of Gaussian noise and for the expected Hilbert-Schmidt norm $ \E \left\| H - \hat H \right\|_{HS}^2$.
Optimality of the rate of convergence  \eqref{eq:main_result2} has to be understood in the minimax sense as classically done in the literature on nonparametric statistics (we refer to \cite{MR2724359} for a detailed introduction to this topic). For simplicity, this optimality result is stated in the case where the domain  $\Omega = [0,1]^{d}$ is the d-dimensional hypercube.  

\begin{theorem} \label{thm:minimax}
Let $H$ be a linear operator belonging to $\mathcal{E}^{r,s}(A_{1},A_{2})$.
Define $q$ by $1/q = 1/r + 1/s$.  
Suppose that the weights in Definition \ref{def:Er} satisfy  $w[k] \leq c_{1} (1+k^2)^{r/d}$  for all $k \in \N$ and some constant $c_{1} > 0$. 
Assume that the PSF locations $y_{1},\ldots,y_{n}$ satisfy the quasi-uniformity condition given in Definition \ref{def:quasi-uniformity}.
Assume that the random values $(\epsilon_i[k])_{i,k}$ in the observation model \eqref{eq:datamodel} are iid Gaussian with zero mean and variance $\sigma^2$.

Then, there exists a constant $c_{0} > 0$ such that
\begin{equation}
\inf_{\hat H } \sup_{H \in \mathcal{E}^{r,s}(A_{1},A_{2})}\E \left\| \hat H - H \right\|_{HS}^2  \geq c_{0}  (\sigma^2 n^{-1})^{\frac{2q}{2q+d}}, \label{eq:lowbound}
\end{equation}
where the above infimum is taken over all possible estimators $\hat H$ (linear integral operators) with SVIR $\hat S \in L^2(\R^d \times \R^d)$ defined as a measurable function.
\end{theorem}
\begin{proof}
 See Section \ref{sec:proofs}.
\end{proof}


\begin{remark}
In this paper we only study the robustness of the method towards perturbations over the discretization of the impulse responses. 
It is also of great interest to study the behavior of the method with respect to  to jitter errors, i.e. what happens if the impulse responses are sampled at perturbed positions $y_i'$ instead of the exact $y_i$? 
This question is left aside in this paper but the analysis in \cite{Feichtinger2004} suggests than one can also prove some robustness of the method with respect to  those type of perturbations.
\end{remark}

\subsection{Illustrations and numerical experiments}

In this section, we highlight the main ideas of the paper through two numerical experiments.

\paragraph{A 1D estimation problem}

In the first experiment, we wish to reconstruct the operator $H$ with kernel $K(x,y)=\frac{1}{\sqrt{|2\pi \Sigma(y)|}} \exp\left( - \langle \Sigma(y)^{-1}(y-x),y-x\rangle\right)$, with diagonal covariance matrices 
$\Sigma(y) = \sigma(y) \Id$ where  $\sigma(y) = 1 + 2\max{(1-y,y)}$ for $y \in [0,1]$.
The SVIR (Space Varying Impulse Response) and the IRC (Impulse Response Coefficients) of this kernel are shown in Fig. \ref{fig:xp_interp1D} (a) and (b). 
Here, we projected the impulse responses on a discrete orthogonal wavelet basis. Notice how the information is compacted,  in (b) compared to (a): most of the information is concentrated on just a fews rows.

In Fig. \ref{fig:xp_interp1D} (c), we show the $7$ impulse responses that are used to estimate the kernel. 
In Fig. \ref{fig:xp_interp1D} (d), we show their projection on the orthogonal wavelet basis. 
The problem studied in this paper is to estimate the SVIR in (a) from the data in (d).
Given the noisy dataset, the proposed algorithm simultaneously interpolates along rows and denoises along columns to obtain the results in Figure \ref{fig:xp_interp1D} (e-h). Notice how the regularization in the vertical direction ($\alpha>0$) allows improving the estimator: the result in (g) is very similar to (a).

 \begin{figure}[htp]
	\centering
	\begin{subfigure}[b]{0.20\textwidth}
		\includegraphics[width=\textwidth]{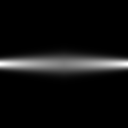}
		\caption{Exact SVIR $S$}
	\end{subfigure}
	\quad	
	\begin{subfigure}[b]{0.2\textwidth}
		\includegraphics[width=\textwidth]{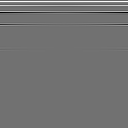}
		\caption{Exact IRC $F$}
	\end{subfigure}
	\quad	
	\begin{subfigure}[b]{0.2\textwidth}
		\includegraphics[width=\textwidth]{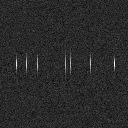}
		\caption{Observed $S(\cdot,y_i)$}
	\end{subfigure}
	\quad	
	\begin{subfigure}[b]{0.2\textwidth}
		\includegraphics[width=\textwidth]{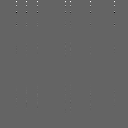}
		\caption{The data set $F_i^\epsilon$}
	\end{subfigure}

	\begin{subfigure}[b]{0.2\textwidth}
		\includegraphics[width=\textwidth]{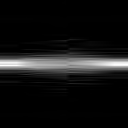}
		\caption{Estimated SVIR $\widehat{S}$-- without denoising ($\alpha = 0$)}
	\end{subfigure}
	\quad
	\begin{subfigure}[b]{0.2\textwidth}
		\includegraphics[width=\textwidth]{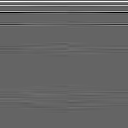}
		\caption{Estimated IRC $\widehat{F}$-- without denoising ($\alpha = 0$)}
	\end{subfigure}
	\quad	
	\begin{subfigure}[b]{0.2\textwidth}
		\includegraphics[width=\textwidth]{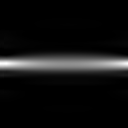}
		\caption{Estimated SVIR $\widehat{S}$-- with denoising ($\alpha = 0.3$)}
	\end{subfigure}
	\quad
	\begin{subfigure}[b]{0.2\textwidth}
		\includegraphics[width=\textwidth]{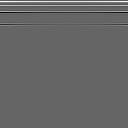}
		\caption{Estimated IRC $\widehat{F }$-- with denoising ($\alpha = 0.3$)}
	\end{subfigure}
	\caption{Illustration of the methodology and its results on a 1D estimation problem.} \label{fig:xp_interp1D}
\end{figure}

\paragraph{A 2D deblurring problem}

In this experiment, we show how the proposed ideas allow estimating a blur operator in imaging and then use this estimate to deblur images. The results are displayed in Fig. \ref{fig:xp_interp2D}.
In Fig. \ref{fig:xp_interp2D} (a), an operator $H$ is applied to a 2D Dirac comb, providing an idea of the operator's shape: each impulse response is an isotropic Gaussian with variance $\sigma(y_1,y_2)$ varying along the vertical direction only (namely $\sigma(y_1,y_2) = 1 + 2\max{(1-y_1,y_1)}$ for $(y_1,y_2) \in [0,1]^2$).
In Fig. \ref{fig:xp_interp2D} (b), we show a set of noisy impulse responses that will be used to perform the estimation. Since the impulse response are near compactly supported, we can isolate each of them in the image to perform the estimation. Here the projection basis $(\phi_k)$ is simply the canonical basis.
In Fig. \ref{fig:xp_interp2D} (c), we show the estimated operator $\widehat{H}$ through its action on a Dirac comb. The estimation seems faithful to the exact operator in (a).

To validate the findings, we perform a deblurring experiment. An sharp image in (d) is blurred with the exact operator $H$ in (a), and some white Gaussian noise is added. Then, using the operator $\widehat{H}$ estimated in (c), we deblur the image with a total variation regularized inverse problem \cite{starck2002deconvolution}. As can be seen, the image is significantly sharper, despite some ringing appearing in the bottom.

 \begin{figure}[htp]
	\centering
	\begin{subfigure}[b]{0.3\textwidth}
		\includegraphics[width=\textwidth]{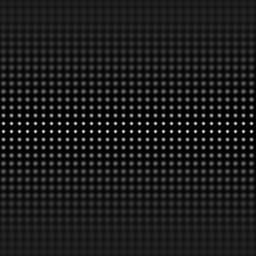}
		\caption{Exact operator}
	\end{subfigure}
	\quad	
	\begin{subfigure}[b]{0.3\textwidth}
		\includegraphics[width=\textwidth]{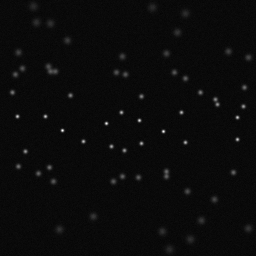}
		\caption{The data set}
	\end{subfigure}
	\quad 	
	\begin{subfigure}[b]{0.3\textwidth}
		\includegraphics[width=\textwidth]{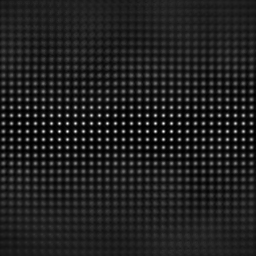}
		\caption{Estimated operator}
	\end{subfigure}
	
	\begin{subfigure}[b]{0.3\textwidth}
		\includegraphics[width=\textwidth]{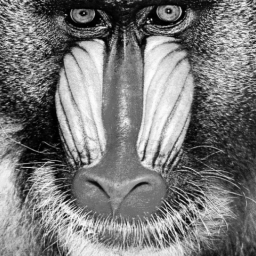}
		\caption{Sharp image \\ $256 \times 256$}
	\end{subfigure}
	\quad	
	\begin{subfigure}[b]{0.3\textwidth}
		\includegraphics[width=\textwidth]{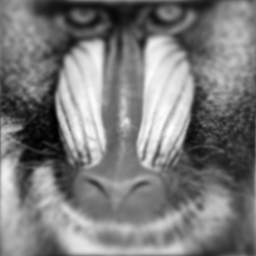}
		\caption{Degraded image \\ pSNR = 19.17dB}
	\end{subfigure}
	\quad 	
	\begin{subfigure}[b]{0.3\textwidth}
		\includegraphics[width=\textwidth]{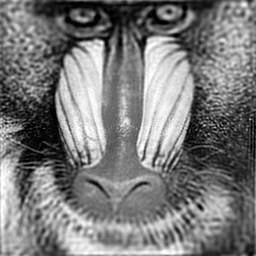}
		\caption{Restored image \\ pSNR = 21.20dB}
	\end{subfigure}	
	
	\caption{
	A 2D estimation used to deblur images.} \label{fig:xp_interp2D}
\end{figure}

\section{Radial basis functions implementation}\label{sec:numerical_algorithm}

The objective of this section is to provide a fast algorithm to solve Problem \eqref{eq:defhatF} and to prove Propositions \ref{prop:complexity1} and \ref{prop:complexity2}. A key observation is provided below.

\begin{lemma}\label{lem:linebyline}
 For $k\in \{1,\hdots, N\}$, the function $\hat F(\cdot)[k]$ is the solution of the following variational problem:
 \begin{equation}\label{eq:main_numerical}
  \min_{f\in H^s(\R^d)} \frac{1}{n} \sum_{i=1}^{n} (F_i^\epsilon[k] - f(y_i))^{2} + \mu\left( \alpha w[k] \|f\|_{L^2(\R^d)}^2 + (1-\alpha) |f|^2_{BL^s(\R^d)}\right).
 \end{equation}
\end{lemma}
\begin{proof}
 It suffices to remark that Problem \eqref{eq:defhatF} consists of solving $N$ independent sub-problems.
\end{proof}

We now focus on the resolution of Sub-problem \eqref{eq:main_numerical} which completely fits the framework of radial basis function approximation.
In the sequel, we gather a few important results related to radial basis functions that will be used to construct the algorithm. 

\subsection{Standard approximation results in RKHS}

A nice way to introduce radial basis functions is through the theory of reproducible kernel Hilbert spaces (RKHS). 
We recall the basic definitions and a few key results regarding RKHS.
Most of them can be found in the book of Wendland \cite{wendland2004scattered}.

\begin{definition}[Positive definite function]
 A continuous function $ \rho :\R^d \to \C$ is called positive semi-definite if, for all $n\in \N$, all sets of pairwise distinct centers $Y=\{y_1,\hdots,y_n\}\subset \R^d$, and all $\alpha \in \C^n$, the quadratic form
 \begin{equation}\label{eq:quadratic}
  \sum_{j=1}^n\sum_{k=1}^n \alpha_j \bar \alpha_k \rho (y_j-y_k)
 \end{equation}
 is nonnegative. It is called positive definite if \eqref{eq:quadratic} is positive for all $\alpha\neq 0$ and all sets of pairwise distinct locations $Y$.
\end{definition}

\begin{definition}[{Reproducing kernel}]
Let  $\mathcal{G}$ denote a Hilbert space of real-valued functions $f:\R^d\to \R$ endowed with a scalar product $\langle \cdot,\cdot\rangle_{\mathcal{G}}$. 
A function $\Phi:\R^d\times \R^d\to \R$ is called reproducing kernel for $\mathcal{G}$ if
\begin{enumerate}
 \item $\Phi(\cdot,y)\in \mathcal{G}, \quad \forall y\in \R^d$,
 \item $f(y) = \langle f, \Phi(\cdot, y)\rangle_{\mathcal{G}}$, for all $f\in \mathcal{G}$ and all $y\in \R^d$.
\end{enumerate}
\end{definition}

\begin{theorem}[RKHS]
Suppose that $\mathcal{G}$ is a Hilbert space of functions $f:\R^d\to \R$. Then the following statements are equivalent:
\begin{enumerate}
 \item the point evaluations functionals $\delta_y$ are continuous for all $y\in \R^d$.
 \item $\mathcal{G}$ has a reproducing kernel.
\end{enumerate}
A Hilbert space satisfying the properties above is called a Reproducing Kernel Hilbert Space (RKHS).
\end{theorem}

The Fourier transform of a function $f \in L^1(\R^d)$ is defined by
\begin{equation}
  \mathcal{F}[f](\xi) = (2 \pi)^{-d/2} \int_{x \in \R^d} f(x) e^{-i \langle x, \xi\rangle}dx,
\end{equation}
and the inverse transform by
\begin{equation}
  \mathcal{F}^{-1}[\widehat{f}](x) = (2 \pi)^{-d/2} \int_{\xi \in \R^d} \widehat{f}(\xi) e^{i \langle x, \xi\rangle }d\xi.
\end{equation}
The Fourier transform can be extended to $L^2(\R^d)$ and to $\mathcal{S}'(\R^d)$ the space of tempered distributions.

\begin{theorem}[{\cite[Theorem 10.12]{wendland2004scattered}}] \label{thm:kernel_hilbert}
 Suppose that $\rho \in C(\R^d) \cap L^1(\R^d)$ is a real-valued positive definite function. Define $\mathcal{G} = \left\{ f \in L^2(\R^d) \cap C(\R^d) : \mathcal{F}[f] / \sqrt{\mathcal{F}[\rho]} \in L^2(\R^d) \right\}$ equipped with
 \begin{equation} \label{eq:inner}
    \langle f,g\rangle_{\mathcal{G}} = (2\pi)^{-d/2} \int_{\R^d} \frac{\mathcal{F}[f](\xi) \overline{\mathcal{F}[g](\xi)}}{ \mathcal{F}[\rho](\xi)} d\xi.
 \end{equation}
Then $\mathcal{G}$ is a real Hilbert space with inner-product $\langle\cdot,\cdot\rangle_{\mathcal{G}}$ and reproducing kernel $\Phi$ defined as $\Phi(x,y) = \rho(x - y)$ for all $x,y \in \R^d$. 
\end{theorem}
This theorem is a consequence of Sobolev embedding theorems \cite{adams2003sobolev}.
In the following, we will make the abuse to call $\rho : \R^d \to \R$ the reproducing kernel of an Hilbert space $\mathcal{G}$. It should be understood as: the reproducing kernel of $\mathcal{G}$ is $\Phi:\R^d \times \R^d \to \R$ defined as $\Phi(x,y) = \rho(x - y)$ for all $x,y \in \R^d$. 

\begin{theorem} \label{theo:RKHSG}
Let $\mathcal{G}$ be an RKHS with positive definite reproducing kernel $\rho:\R^d \to \R$.
Let $(y_1,\hdots,y_n)$ denote a set of points in $\R^d$ and $z\in \R^n$ denote a set of altitudes.
The solution of the following approximation problem
 \begin{equation}
  \min_{u\in \mathcal{G}} \frac{1}{n}\sum_{i=1}^n (u(y_i) - z[i])^2 + \frac{\mu}{2} \|u\|_{\mathcal{G}}^2
 \end{equation}
 can be written as:
 \begin{equation}\label{eq:reconstructu}
  u(x) = \sum_{i=1}^n c[i] \rho(x-y_i),
 \end{equation}
 where vector $c\in \R^n$ is the unique solution of the following linear system of equations
 \begin{equation}\label{eq:defG}
  (G + n\mu \Id)c = z \mbox{ with } G[i,j]=\rho(y_i-y_j).
 \end{equation}
\end{theorem}

It is shown in \cite{wendland2004scattered}, that the condition number of $G$ depends on the ratio $h_{Y,\Omega}/q_{Y,\Omega}$.
For numerical reasons it might therefore be useful to implement thinning methods in order to discard locations that are too close to each other without creating larger gaps, if possible \cite{dyn2002adaptive,iske2004multiresolution,dyn2008meshfree}.

\subsection{Application to our problem}

Let us now show how the above results help solving Problem \eqref{eq:main_numerical}.

\begin{proposition} \label{prop:rbf}
  Let $\mathcal{G}_k$ be the Hilbert space of functions $f:\R^d \to \R$ such that $|f|_{BL^s(\R^d)}^2 + \| f\|_{L^2(\R^d)}^2 < + \infty$, equipped with the inner product:
  \begin{equation}
    \langle f,g\rangle_{\mathcal{G}_k} = (1-\alpha) \left\langle f,g\right\rangle_{BL^s(\R^d)} + \alpha w[k] \langle f, g\rangle_{L^2(\R^d)}^2.
  \end{equation}
Then $\mathcal{G}_k$ is an RKHS and its scalar product reads
\begin{equation}
\langle f,g\rangle_{\mathcal{G}_k} = (2\pi)^{-d/2} \int_{\R^d} \frac{\mathcal{F}[f](\xi) \overline{\mathcal{F}[g](\xi)}}{ \mathcal{F}[\rho_k](\xi)} d\xi,
\end{equation}
where the reproducing kernel $\rho_k$, is defined by:
\begin{equation} \label{eq:kernel_rbf_definition}
 \mathcal{F}[\rho_k](\xi) = \left((1-\alpha)\|\xi\|^{2s} + \alpha w[k] \right)^{-1}.
\end{equation} 
\end{proposition}
\begin{proof}
 The proof is a direct application of the different results stated previously.
\end{proof}

The Fourier transform $\mathcal{F}[\rho_{k}]$ is radial, so that $\rho_{k}$ is radial too and the resolution of \eqref{eq:main_numerical} fits the formalism of radial basis functions interpolation/approximation \cite{buhmann2003radial}.

\begin{remark}
For some applications, it makes sense to set $w[k]=0$ for some values of $k$. 
For instance, if $(\phi_k)_{k\in \N}$ is a wavelet basis, then it is usually good to set $w[k]=0$ when $k$ is the index of a scaling wavelet.
In that case, the theory of conditionally positive definite kernels should be used instead of the one above. 
We do not detail this aspect since it is well described in standard textbooks \cite{wendland2004scattered,buhmann2003radial}.
\end{remark}

The whole procedure computing $\widehat{F}$ is presented in Algorithm \ref{alg:FISTAP}. 
The principle of the algorithm is derived from Lemma \ref{lem:linebyline} showing that computing $\widehat{F}$ solution of \eqref{eq:defhatF}, boils down to solving $N$ independent sub-systems. Each sub-system computes $\widehat{F}(\cdot)[k]$ and according to Proposition \ref{prop:rbf} it falls in the formalism of RKHS with an explicit definition of the kernel. Therefore, in virtue of Theorem \ref{theo:RKHSG} each function $\widehat{F}(\cdot)[k]$ can be computed by solving a $n \times n$ linear system. 
The resolution of the linear systems can accelerated using LU decompositions. Hence, it starts with a preprocessing step.

The associated $\widehat{S}$ can be recovered for all $(x,y) \in \R^d \times \R^d$ through
\begin{equation}
 \widehat{S}(x,y) = \sum_{k=1}^N \widehat{F}(y)[k] \phi_k(x) = \sum_{k=1}^N \sum_{i=1}^n c_{k}[i] \rho_k(y - y_i) \phi_k(x).
\end{equation}
Before being able to use $\widehat{S}$ for subsequent numerical algorithms, the IRC $\widehat{F}$ might have to be discretized or sampled. The complexity of this step is not comprised in Proposition \ref{prop:complexity1} and depends on the discretization procedure. 
In many cases, $\widehat{F}(y)[k]$ has to be evaluated on a Cartesian grid. This step can be performed efficiently by using nonuniform fast Fourier transforms or multipole methods \cite{wendland2004scattered}.

\begin{algorithm}
\caption{Computation of $\widehat{F}$}
\label{alg:FISTAP}
\begin{algorithmic}[1]
\INPUT $\quad$
\begin{itemize} \itemsep0em 
 \item[] Weight vector $w\in \R^N$
 \item[] Regularity $s\in \N$
 \item[] PSF locations $Y=\{y_1,\hdots, y_n\}\in \R^{d\times n}$
 \item[] Observed data $(F_i^\epsilon)_{1\leq i \leq n}$, where $F_i^\epsilon\in \R^N$
\end{itemize}

\OUTPUT $\quad$
  \begin{itemize} \itemsep0em
   \item[] The IRC estimator $\widehat{F}$ 
  \end{itemize}
\BEGIN
\State Identify the $m\leq N$ weights of identical values in vector $w\in \R^N$. \Comment{$O(N)$}
\For{Each unique weight $\omega$} \Comment{$O(mn^3)$}
\State Compute matrix $G$ from formula \eqref{eq:defG} with $\rho_\omega$ defined in \eqref{eq:kernel_rbf_definition}.
\State Compute a LU decomposition of $M_\omega = (G+n\mu \Id) = L_\omega U_\omega$.
\EndFor
\For{$k=1$ to $N$} \Comment{$O(N n^2)$}
\State Identify the value $\omega$ such that $w[k]=\omega$.
\State Set $z=(F_i^\epsilon[k])_{1\leq i \leq n}$.
\State Solve the linear system $L_\omega U_\omega c_k = z$.
\State Possibly reconstruct $\hat F$ by (see equation \eqref{eq:reconstructu})
\begin{equation*}
\hat F(y)[k] = \sum_{i=1}^n c_k[i] \rho_\omega(y-y_i).
\end{equation*}
\EndFor
\end{algorithmic}
\end{algorithm}

\section{Proofs of the main results}
\label{sec:proofs}

First we prove Theorem \ref{thm:main_result} about the convergence rate of the quadratic risk $\E \| H - \hat{H} \|_{HS}^2$. 

\subsection{Operator norm risk}

To analyse the theoretical properties of a given estimator of the operator $H$, we introduce the quadratic risk defined as:
\begin{equation} \label{eq:risk_operator}
 R(\hat H, H) = \E \left\| \hat H - H \right\|_{HS}^2,
\end{equation}
where $\hat H$ is the operator associated to the SVIR $\hat{S}$ defined in \eqref{eq:hat_T}. The above expectation is taken with respect to the distribution of the observations in \eqref{eq:datamodel}. 
Notice that $\| H \|_{HS} = \| K \|_{L^2(\R^d\times \Omega)} = \| S \|_{L^2(\R^d \times \Omega)}$. From this observation we get that:
\begin{align}
 R(\hat H, H) &= \E \left\| \hat H - H \right\|_{HS}^2 \nonumber \\
 & \leq 2\left( \left\| H - H_N \right\|_{HS}^2 +  \E \left\| H_N - \hat{H} \right\|_{HS}^2  \right) \nonumber  \\
 & = 2\left( \underbrace{\left\| S - S_N \right\|_{L^2(\R^d \times \Omega)}^2}_{\epsilon_d(N)} +  \underbrace{\E \left\| S_N - \hat{S} \right\|_{L^2(\R^d \times \Omega)}^2}_{\epsilon_{e}(n)}  \right) \label{eq:risk_separation},
\end{align}
where $H_N$ is the operator associated to the SVIR $S_N$ defined by $S_N(x,y) = \sum_{k=1}^N F(y)[k] \phi_{k}(x)$ and $\hat{H}$ the estimating operator associated to the SVIR $\hat{S}$ as in \eqref{eq:hat_T}. 

In equation \eqref{eq:risk_separation}, the risk is decomposed as the sum of two terms $\epsilon_e(n)$ and $\epsilon_d(N)$ (standard bias/variance decomposition in statistics). The first one $\epsilon_d(N)$ is the error introduced by the discretization step.
The second term $\epsilon_e(N)$ is the quadratic risk between $S_N$ and the estimator $\hat{S}$.
In the next sections, we provide upper-bounds for $\epsilon_d(N)$ and $\epsilon_e(n)$.

\subsection{Discretization error \texorpdfstring{$\epsilon_d$}{}}

The discretization error $\epsilon_d(N)$ can be controlled using the standard approximation result below (see e.g. \cite[Theorem 9.1, p. 503]{mallat1999wavelet}).
\begin{theorem}\label{thm:mallat}
There exists a universal constant $c>0$ such that for all $f \in \mathcal{E}^r(\R^d)$ the following estimate holds
\begin{equation}
 \|f-f_N\|_2^2 \leq c \|f\|_{\mathcal{E}^r(\R^d)}^2 N^{-2r/d},
\end{equation}
with $f_N=\sum_{k=1}^N \langle f,\phi_k\rangle \phi_k$.
\end{theorem}

\begin{corollary}\label{cor:epsilond}
  Under the assumption $H \in \mathcal{E}^{r,s}(A_1,A_2)$, the discretization error satisfies:
 \begin{equation}
  \epsilon_d(N) \lesssim N^{-2r/d}.
 \end{equation}
\end{corollary}
\begin{proof}
By assumption \eqref{eq:regIR}, $S(\cdot,y)\in \mathcal{E}^r(\R^d)$ for almost every $y\in \Omega$.
Therefore, by Theorem \ref{thm:mallat}:
 \begin{align}
  \|S(\cdot,y)-S_N(\cdot,y) \|_{L^2(\R^d)}^2 \leq cN^{-2r/d} \|S(\cdot,y) \|_{\mathcal{E}^r(\R^d)}^2
 \end{align}
 Finally:
 \begin{align*}
  \|S-S_N\|_{L^2(\R^d \times \Omega)}^2 &= \int_{y\in \Omega} \|S(\cdot,y) - S_N(\cdot,y)\|_{L^2(\R^d)}^2 \,dy\\ 
  &\leq c \left(\int_{y\in \Omega} \|S(\cdot,y)\|_{\mathcal{E}^r(\R^d)}^2 dy \right) N^{-2r/d} \\
  & \leq c A_2 N^{-2r/d} \\
 \end{align*}
\end{proof}

\subsection{Estimation error \texorpdfstring{$\epsilon_e$}{}}

This section provides an upper-bound on the estimation error
\begin{equation}
\epsilon_e(n) = \E \left\| S_N - \hat{S} \right\|_{L^2(\R^d \times \Omega)}^2.
\end{equation}
This part is significantly harder than the rest of the paper.
Let us begin with a simple remark.
\begin{lemma}
 The estimation error satisfies 
 \begin{equation}
  \epsilon_e(n) = \E \left\| F - \hat{F} \right\|_{\R^N \times L^2(\Omega)}^2.
 \end{equation}
\end{lemma}
\begin{proof}
Since $(\phi_{k})_{1 \leq k \leq N}$ is an orthonormal basis, Parseval's theorem gives
\begin{align}
 \| S_N - \hat{S} \|_{L^2(\R^d \times \Omega)}^2 &= \int_{\Omega}\int_{\R^d} \left( S_N(x,y) - \hat{S}(x,y) \right)^2 dx dy \nonumber \\
 & = \int_{\Omega}\int_{\R^d} \left( \sum_{k=1}^N (F(y)[k] - \hat{F}(y)[k])\phi_k(x)\right)^2 dx dy \nonumber \\
 & = \int_{\Omega} \sum_{k=1}^N (F(y)[k] - \hat{F}(y)[k])^2 dy  \nonumber \\
 & = \sum_{k=1}^N \| F(\cdot)[k] - \hat{F}(\cdot)[k] \|_{L^2(\Omega)}^2  =: \| F - \hat{F} \|_{\R^N \times L^2(\Omega)}^2. \label{eq:risk}
\end{align}
\end{proof}

By Lemma \ref{lem:linebyline} the estimator defined in \eqref{eq:defhatF} can be decomposed as $N$ independent estimators. 
Lemma \ref{lem:cv_rate_tvir_wavelet_row} below provides a convergence rate for each of them. 
This result is strongly related to the work in \cite{utreras1988convergence} on smoothing splines. 
Unfortunately, we cannot directly apply the results in  \cite{utreras1988convergence} to our setting since the kernel defined in  \eqref{eq:kernel_rbf_definition} is not that of a thin-plate smoothing spline.

\begin{lemma} \label{lem:cv_rate_tvir_wavelet_row}
 Suppose that $\Omega \subset \R^d$ is a bounded connected open set in $\R^d$ with Lipschitz continuous boundary.
 Let $Y = \{ y_1, \ldots, y_n\} \subset \Omega$ be a quasi-uniform sampling set of PSF locations.
 Recall that $\| f \|_{\mathcal{G}_k(\R^d)} = (1-\alpha) |f|_{BL^s(\R^d)} + \alpha w[k] \| f \|_{L^2(\R^d)}$, for all $k \geq 1$.
 Then, each function $\hat{F}(\cdot)[k]$ solution of Problem \eqref{eq:main_numerical} satisfies:
 \begin{equation}
  \E \| \hat{F}(\cdot)[k] - F(\cdot)[k] \|_{L^2(\Omega)}^2 \lesssim \mu (1-\alpha)^{-1} \| F(\cdot)[k] \|_{\mathcal{G}_k(\R^d)}^2 + n^{-1} \sigma^2 \left[ (1-\alpha) \mu\right]^{-\frac{d}{2s}} (1-\alpha)^{-1}, \label{eq:upperbound1}
 \end{equation}
 provided that $n \mu^{d/2s} \geq 1$.
\end{lemma}

\begin{proof}
 In order to prove the upper-bound \eqref{eq:upperbound1}, we first decompose the expected squared error  $\E \|  \hat{F}(\cdot)[k] - F(\cdot)[k] \|_{L^2(\Omega)}^2$ into bias and variance terms:
 \begin{equation}
    \E \| \hat{F}(\cdot)[k] - F(\cdot)[k] \|_{L^2(\Omega)}^2 \leq 2\left( \underbrace{\| \hat{F}^0(\cdot)[k] - F(\cdot)[k] \|_{L^2(\Omega)}^2}_{\textrm{Bias term}} + \underbrace{\E \| \hat{F}^0(\cdot)[k] - \hat{F}(\cdot)[k] \|_{L^2(\Omega)}^2 }_{\textrm{Variance term}}\right),
 \end{equation}
 where $\hat{F}^0(\cdot)[k]$ is the solution of the noise-free problem
 \begin{equation} 
  \hat{F}^0(\cdot)[k] = \argmin_{f \in H^s(\R^d)} \frac{1}{n} \sum_{i=1}^{n} (F(y_i)[k] - f(y_i))^{2} + \mu \left( \alpha w[k] \|f\|_{L^2(\R^d)}^2 + (1-\alpha) |f|^2_{BL^s(\R^d)}\right).
\end{equation}
We then treat the bias and variance terms separately. 

\paragraph{Control of the bias}

The bias control relies on sampling inequalities in Sobolev spaces. 
They first appeared in \cite{duchon1978erreur} to control the norm of functions in Sobolev spaces with scattered zeros. 
They have been generalized in different ways, see e.g. \cite{wendland2005approximate} and \cite{arcangeli2007extension}.  In this paper, we will use the following result from \cite{arcangeli2007extension}.

\begin{theorem}[{\cite[Theorem 4.1]{arcangeli2007extension}}] \label{thm:sampling_ineq_sobolev}
 Let $\Omega \subset \R^d$ be a bounded connected open set with Lipschitz continuous boundary, 
 and $p,q,x \in [1,+\infty]$ be given. Let $s$ be a real number such that $s \geq d$ if $p = 1$, $s > d/p$ if $1 < p < \infty$ or $s \in \N^*$ if $p = \infty$. Furthermore, let $l_0 = s - d(1/p - 1/q)_+$ and $\gamma = \max(p,q,x)$ where $(\cdot)_+ = \max(0,\cdot)$. 
 
 Then there exist two positive constants $\eta_s$ (depending on $\Omega$ and $s$) and $C$ (depending on $\Omega$, $n, s, p, q$ and $x$) satisfying the following property: for any finite set $Y \subset \bar{\Omega}$ (or $Y \subset \Omega$ if $p=1$ and $s=d$) such that $h_{Y,\Omega} \leq \eta_s$, for any $u \in W^{s,p}(\Omega)$ and for any $l = 0, \ldots, \lceil l_0 \rceil -1$, we have
 \begin{equation}
    \| u \|_{W^{l,q}(\Omega)} \leq C \left( h_{Y, \Omega}^{s - l - d(1/p - 1/q)_+} |u|_{W^{s,p}(\Omega)} + h_{Y, \Omega}^{d/\gamma - l} \| u|_{Y} \|_x \right),
 \end{equation}
 where $\| u|_{Y} \|_x = \left( \sum_{i=1}^n u(y_i)^x \right)^{1/x}$. If $s \in \N^*$ this bound also holds with $l = l_0$ when either $p < q < \infty$ and $l_0 \in \N$ or $(p,q) = (1,\infty)$ or $p \geq q$.
\end{theorem}

The above theorem is the key to obtain Proposition \ref{prop:cv_rate_function_kernel_noise_free} below. 


\begin{proposition} \label{prop:cv_rate_function_kernel_noise_free}
  Set $0 \leq \alpha < 1$ and let $\mathcal{G}_k(\Omega)$ be the RKHS with norm defined by $\| \cdot\|_{\mathcal{G}_k(\Omega)}^2 = (1-\alpha) |\cdot|_{BL^s(\Omega)}^2 + \alpha w[k] \| \cdot \|_{L^2(\Omega)}^2$.
  Let $u \in H^s(\Omega)$ denote a target function and $Y = \{ y_1, \ldots, y_n \} \subset \Omega$ a data site set.
  Let $f_\mu$ denote the solution of the following variational problem
  \begin{equation}\label{eq:deffmu}
   f_\mu = \argmin_{f \in \mathcal{G}(\R^d)} \frac{1}{n} \sum_{i=1}^n ( u(y_j) - f(y_j) )^2 + \mu \|  f \|_{\mathcal{G}_k(\R^d)}^2.
  \end{equation}
  Then
  \begin{equation}
    \| f_\mu -  u \|_{L^2(\Omega)} \leq C\left( (1-\alpha)^{-1/2} h_{Y,\Omega}^s  + h_{Y,\Omega}^{d/2} \sqrt{ n \mu} \right) \| u \|_{\mathcal{G}_k(\R^d)},
  \end{equation}
  where $C$ is a constant depending only on $\Omega$ and $s$ and $h_{Y,\Omega}$ is the fill distance defined in \ref{def:fill_dist}.
\end{proposition}
\begin{proof}
By applying the Sobolev sampling inequality of Theorem \ref{thm:sampling_ineq_sobolev} for $p = q = x = 2$, $l = 0$, we get
 \begin{equation}
  \| v \|_{L^2(\Omega)} \leq C \left( h_{Y,\Omega}^s |v|_{H^s(\Omega)} + h_{Y, \Omega}^{d/2} \left(\sum_{i=1}^n v(y_i)^2 \right)^{1/2} \right),
 \end{equation}
 for all  $v\in H^s$.
 This inequality applied to function $v=f_\mu-u$ yields
 \begin{equation}\label{eq:fmu_uSamplingInequality}
  \| f_\mu-u \|_{L^2(\Omega)} \leq C \left( h_{Y,\Omega}^s |f_\mu-u|_{H^s(\Omega)} + h_{Y, \Omega}^{d/2} \left(\sum_{i=1}^n (f_\mu(y_i)-u(y_i)^2 \right)^{1/2} \right).
 \end{equation}

The remaining task is to bound $|f_\mu-u|_{H^s(\Omega)}$ and $\left(\sum_{i=1}^n (f_\mu(y_i)-u(y_i))^2 \right)^{1/2}$ by $\|u\|_{\mathcal{G}_k(\R^d)}$. To this end, let us define two functionals $f \mapsto E(f) = \frac{1}{n} \sum_{i=1}^n ( u(y_j) - f(y_j) )^2$ and $f \mapsto J(f) =  \|  f \|_{\mathcal{G}_k(\R^d)}^2$. 
We will use the function $f_0 : \Omega \to \R$ defined as the solution of
\begin{equation}
 f_0 = \argmin_{ \substack{f \in \mathcal{G}_k(\R^d) \\ f(y_j) = u(y_j) } } \|  f \|_{\mathcal{G}_k(\R^d)}^2,
\end{equation}
We notice that the set $\{ f \in \mathcal{G}_k(\R^d) \, | \, \forall 1 \leq j \leq n, \, f(y_j) = u(y_j)  \}$ is non-empty, convex and closed. Furthermore, the squared norm $\|\cdot\|^2_{\mathcal{G}_k(\R^d)}$ is strictly convex. Those two facts imply that the function $f_0$ is uniquely determined. 

Since $f_\mu$ is the minimizer of \eqref{eq:deffmu}, it satisfies 
\begin{equation}
E(f_\mu) + \mu J(f_\mu) \leq E(f_0) + \mu J(f_0). 
\end{equation}
In addition $E(f_0) = 0$ and $J(f_0) \leq  J(u)$. Therefore we have the following sequence of inequalities:
\begin{equation}
    E(f_\mu) + \mu J(f_\mu) \leq E(f_0) + \mu J(f_0) = \mu J(f_0) \leq \mu J(u).
\end{equation}
Hence, 
\begin{equation}\label{eq:boundEfmu}
  E(f_\mu) = \frac{1}{n} \sum_{i=1}^n ( u(y_j) - f_\mu(y_j) )^2 \leq \mu \| u \|_{\mathcal{G}_k(\R^d)}^2.
\end{equation}

To finish, the triangle inequality yields $|f_\mu-u|_{H^s(\Omega)}\leq |f_\mu|_{H^s(\Omega)} + |u|_{H^s(\Omega)}$. The equivalence between the Sobolev semi-norm $|\cdot|_{H^s}$ and the Beppo-Levi semi-norm $|\cdot|_{BL^s}$ yields
\begin{equation} \label{eq:boundfmu_u}
  \begin{split}
 (1-\alpha)|f_\mu|_{H^s(\Omega)}^2  & \lesssim \|f_\mu\|_{\mathcal{G}_k(\R^d)}^2 \leq \| u \|_{\mathcal{G}_k(\R^d)}^2, \\
 (1-\alpha)|u|_{H^s(\Omega)}^2  & \lesssim \| u \|_{\mathcal{G}_k(\R^d)}^2.
  \end{split}
\end{equation}
Replacing bounds \eqref{eq:boundfmu_u} and \eqref{eq:boundEfmu} in the sampling inequality \eqref{eq:fmu_uSamplingInequality} completes the proof of Proposition \ref{prop:cv_rate_function_kernel_noise_free}.
\end{proof}


Applying Proposition \ref{prop:cv_rate_function_kernel_noise_free} to $\hat{F}^0(\cdot)[k]$, we get
\begin{equation}
  \| \hat{F}^0(\cdot)[k] - F(\cdot)[k] \|_{L^2(\Omega)}^2 \leq C\left(  (1-\alpha)^{-1/2} h_{Y,\Omega}^s + \sqrt{\mu n} h_{Y,\Omega}^{d/2} \right)^2 \|F(\cdot)[k]\|_{\mathcal{G}_k(\R^d)}^2.
\end{equation} 

The trick is now to use the quasi-uniformity condition to control $h_{Y,\Omega}^s$ and 
$\sqrt{\mu n} h_{Y,\Omega}^{d/2}$. This is achieved using the following proposition.
\begin{proposition}[{\cite[Proposition 14.1]{wendland2004scattered} or \cite{utreras1988convergence}}] \label{prop:equiv_fill_distance}
 Let $Y = \{ y_1, \ldots, y_n \} \subset \Omega$ be a quasi-uniform set with respect to $B$. Then, there exist constants $c > 0$ and $C > 0$ depending only on $d$, $\Omega$ and $B$ such that,
 \begin{equation}\label{eq:conditionhy}
  c n^{-1} \leq h_{Y, \Omega}^d \leq C n^{-1}.
 \end{equation}
\end{proposition}
Condition $n \mu^{d/2s} \geq 1$ combined with the right-hand-side of \eqref{eq:conditionhy} yields $h_{Y,\Omega}^d \leq C \mu^{d/2s}$, so that $h_{Y,\Omega}^s \lesssim \sqrt{\mu}$. Similarly, the right-hand-side of \eqref{eq:conditionhy} yields $\sqrt{\mu n} h_{Y,\Omega}^{d/2}\lesssim \sqrt{\mu}$.
Hence
\begin{equation}
  \begin{split}
  \| \hat{F}^0(\cdot)[k] - F(\cdot)[k] \|_{L^2(\Omega)}^2 & \lesssim (1-\alpha)^{-1} \mu \| F(\cdot)[k]\|_{\mathcal{G}_k(\R^d)}^2.
  \end{split}
\end{equation}

\paragraph{Control of the variance}

The variance term is treated following arguments similar to those in \cite{utreras1988convergence}. However, the change of kernel  needs additional treatment. First of all, note that due to the linearity of the estimators of Problem \eqref{eq:main_numerical} (that can be seen from equation \eqref{eq:defG}), we have $\hat{F}_{\mu}^0(\cdot)[k] - \hat{F}_{\mu}(\cdot)[k] = f_{k}^{\eta}$ with $\eta \in \R^n$ defined as $\eta[i] = \epsilon_i[k]$ and
\begin{equation} \label{eq:noise_smoothing}
 f_{k}^{\eta} =  \argmin_{f \in H^s(\R^d)} \frac{1}{n} \sum_{i=1}^{n} \left( f(y_i) - \eta[i] \right)^2 + \mu \left( \alpha w[k] \|f\|_{L^2(\R^d)}^2 + (1-\alpha) |f|^2_{BL^s(\R^d)}\right).
\end{equation}
We therefore need to estimate $\E \| f_{k}^{\eta}\|_{L^2(\Omega)}^2$. From Theorem \ref{thm:sampling_ineq_sobolev} applied with $p = q = x = 2$ and $l = 0$ we obtain that for $u \in H^{s}(\Omega)$
$$
 \| u \|_{L^2(\Omega)} \leq C \left( h_{Y,\Omega}^s |u|_{H^s(\Omega)} + h_{Y,\Omega}^{d/2} \| u|_Y \|_2\right).
$$
Using the above inequality together with Proposition \ref{prop:equiv_fill_distance}, we get that
$$
 \|f_{k}^{\eta} \|_{L^2(\Omega)}^2 \leq 2 C \left( h_{Y,\Omega}^{2s} | f_{k}^{\eta} |_{H^s(\Omega)}^2 + n^{-1} \sum_{i=1}^n f_{k}^{\eta}(y_i)^2 \right).
$$
As in \cite{utreras1988convergence}, let us define the $n \times n$ symmetric matrix $ \tilde{\Gamma}$ such that
\begin{equation} \label{eq:gamma_tilde_matrix}
 \langle \tilde{\Gamma} z, z \rangle = \min_{ \substack{u \in BL^s(\R^d) \\ u(y_i) = z[i] }} (1 - \alpha) | u |_{H^s(\R^d)}^2 + \alpha w[k] \|u\|_{L^2(\R^d)}^2.
\end{equation}
The solution of Problem \eqref{eq:gamma_tilde_matrix} is a spline interpolating the data $(y_i, z[i])_{i=1}^n$. Using this matrix, we can write \eqref{eq:noise_smoothing} as:
$$
  \min_{z \in \R^n} \frac{1}{n} \sum_{i=1}^n (z[i] - \eta[i])^2 + \mu \langle \tilde{\Gamma} z, z \rangle,
$$
see \cite{utreras1988convergence,utreras1979cross,wahba1979convergence} for details. Thus, the solution $\hat{z} = (f_{k}^{\eta}(y_i))_{i=1}^n$ is obtained by:
$$
 \hat{z} = ( \Id + n\mu \tilde{\Gamma})^{-1} \eta.
$$
By letting $E_\mu = (\Id + n\mu \tilde{\Gamma})^{-1}$, we obtain
$$
  n^{-1} \sum_{i=1}^n f_{k}^{\eta}(y_i)^2 = n^{-1} \sum_{i=1}^n \hat{z}[i]^2 = n^{-1} \eta^{T} E_\mu^2 \eta
$$
and
\begin{align*}
 (1-\alpha) | f_{k}^{\eta} |_{H^s(\R^d)}^2 + \alpha w[k] \|f_{k}^{\eta}\|_{L^2(\R^d)}^2 &= \hat{z}^T \tilde{\Gamma} \hat{z} = \eta^{T} E_\mu \tilde{\Gamma} E_\mu \eta \\
 & = (n\mu)^{-1}\eta^{T} E_\mu ( E_\mu^{-1} - \Id) E_\mu  \eta  \\
 & = (n\mu)^{-1}\eta^{T} (E_\mu - E_\mu^2) \eta.
\end{align*}
Thus
$$
  | f_{k}^{\eta} |_{H^s(\Omega)}^2 \leq | f_{k}^{\eta} |_{H^s(\R^d)}^2 \leq (n\mu(1 -\alpha))^{-1}\eta^{T} ( E_\mu - E_\mu^2) \eta.
$$
Using the fact that $\eta$ has i.i.d. components with zero mean and variance $\sigma^2$, we get that,
$$
 \E \left[ n^{-1} \sum_{i=1}^n f_{\lambda}^{\eta}(y_i)^2 \right] = n^{-1} \sigma^2 \textrm{Tr}(E_\mu^2),
$$
and on the other hand
\begin{align*}
 \E | f_{k}^{\eta} |_{H^s(\Omega)}^2 &\leq (n\mu (1 -\alpha))^{-1} \sigma^2 (\textrm{Tr}(E_\mu) - \textrm{Tr}(E_\mu^2))\\
 & \leq (n\mu(1 -\alpha))^{-1} \sigma^2 \textrm{Tr}(E_\mu).
\end{align*}

We now have to focus on the estimation of both $\textrm{Tr}(E_\mu) = \sum_{i=1}^n (1 + n \mu \lambda_{i}(\tilde{\Gamma}) )^{-1}$ and $\textrm{Tr}(E_\mu^2) = \sum_{i=1}^n (1 + n \mu \lambda_{i}(\tilde{\Gamma}) )^{-2}$, where $\lambda_i(\tilde{\Gamma})$ is the $i$-th eigenvalue of $\tilde{\Gamma}$. This will be achieved by analyzing the eigenvalues of the matrix $\tilde{\Gamma}$. 
This step is quite cumbersome. Fortunately, we can rely on the work of Utreras who analyzed the eigenvalues of the matrix $\Gamma$ associated to thin-plate splines in \cite{utreras1988convergence}. 
Matrix $\Gamma$ is defined in a similar way as \eqref{eq:gamma_tilde_matrix}:
\begin{equation} \label{eq:gamma_matrix}
 \langle \Gamma z, z \rangle = \min_{ \substack{u \in BL^s(\R^d) \\ u(y_i) = z[i] }} | u |_{H^s(\R^d)}^2.
\end{equation}
One therefore has that $(1-\alpha) z^T \Gamma z \leq z^T \tilde{\Gamma} z$ for all $z \in \R^N$. Therefore the matrix $\tilde{\Gamma} - (1-\alpha)\Gamma$ is semi-definite positive. By virtue of Weyl Monotonicity Theorem  \cite{weyl1912asymptotische}, we get that $(1-\alpha)\lambda_{i}(\Gamma) \leq \lambda_{i}(\tilde{\Gamma})$. Hence we can bound the traces of the matrices $E_\mu$ and $E_\mu^2$ as follows
\begin{align*}
 \textrm{Tr}(E_\mu) & \leq \sum_{i=1}^n (1 + (1-\alpha)n \mu \lambda_{i}(\Gamma) )^{-1},\\
 \textrm{Tr}(E_\mu^2) & \leq \sum_{i=1}^n (1 + (1-\alpha)n \mu \lambda_{i}(\Gamma) )^{-2}. \\
\end{align*}
It is shown in \cite{utreras1988convergence}, that $\gamma = \begin{pmatrix} s-1 + d \\ s-1\end{pmatrix}$ eigenvalues $\lambda_{i}(\Gamma)$ are null and the others satisfy $i^{2s/d} n^{-1} \lesssim \lambda_{i}(\Gamma) \lesssim i^{2s/d} n^{-1}$ for $\gamma+1 \leq i \leq n$. Following \cite{utreras1988convergence}, it can be shown that both traces are bounded by quantities proportional to $\left[ (1-\alpha) \mu\right]^{-d/2s}$. Thus one has that
$$
 \E\|f_{k}^{\eta} \|_{L^2(\Omega)}^2 \lesssim \sigma^2( n^{-1}\left[ (1-\alpha) \mu\right]^{-d/2s} + n^{-1} h_{Y,\Omega}^{2s} \mu^{-1} \left[ (1-\alpha) \mu\right]^{1-d/2s}).
$$
Since $\mu^{d/2s} n \geq 1$ and using Proposition \ref{prop:equiv_fill_distance} that gives $n \lesssim h_{Y,\Omega}^{-d}$ we obtain that $h_{Y,\Omega}^{2s} \mu^{-1} \lesssim 1$. Hence
$$
 \E\|f_{k}^{\eta} \|_{L^2(\Omega)}^2  \lesssim \sigma^2 n^{-1}\left[ (1-\alpha) \mu\right]^{-d/2s}\left( 1 + h_{Y,\Omega}^{2s} [(1-\alpha)\mu]^{-1}\right) \lesssim \sigma^2 n^{-1}\left[ (1-\alpha) \mu\right]^{-d/2s} (1-\alpha)^{-1},
$$
which completes the proof of Lemma \ref{lem:cv_rate_tvir_wavelet_row}.
\end{proof}

Finally, we will need the following technical Lemma.

\begin{lemma} \label{lem:invert}
  Let $H$ be an operator in $\mathcal{E}^{r,s}(A_1,A_2)$ with SVIR $S$ \eqref{eq:defTVIR} and IRC $F$ \eqref{eq:IRC}. We have
  \begin{align}
    \partial_y^\alpha \langle S(\cdot,y), \phi_k \rangle &= \langle \partial_y^\alpha S(\cdot,y), \phi_k \rangle \quad \forall |\alpha| \leq s, \textrm{and for a.e. } y \in \Omega, \\
    F(\cdot)[k] : y \mapsto \langle S(\cdot,y) , \phi_k \rangle & \in H^s(\R^d) \quad \forall k \in \N, \\
    \sum_{k \in \N} | F(\cdot)[k]|_{BL^s(\R^d)}^2 &= \int_{x \in \R^d} |S(x,\cdot)|_{BL^s(\R^d)}^2 dx, \\
    \sum_{k \in \N} w[k] \left\| F(\cdot)[k] \right\|_{L^2(\R^d)}^2 &= \int_{\R^d} \left\| S(\cdot, y) \right\|_{\mathcal{E}^r(\R^d)}^2 dy.
  \end{align}
\end{lemma}

\begin{proof}
 
 The first point is derived using a result in \cite[Theorem 7.40]{jones1982theory}: since $S \in L^2(\R^d \times \R^d)$, it defines a generalized function. We obtain that $\partial_y^\alpha \langle S(\cdot,y), \phi_k \rangle = \langle \partial_y^\alpha S(\cdot,y), \phi_k \rangle$ in the sense of generalized functions. Moreover,
 \begin{equation*}
  \begin{split}
    \int_{\R^d} \left| \langle \partial_y^\alpha S(\cdot,y), \phi_k \rangle \right|^2 dy & \leq \int_{\R^d} \left\| \partial_y^\alpha S(\cdot,y) \right\|_{L^2(\R^d)}^2 dy \\ 
    & \leq A_1,
   \end{split}
 \end{equation*}
 since $S \in \mathcal{E}^{r,s}(A_1,A_2)$. 
 Thus the equality is also valid in $L^2(\R^d)$ and $ \partial_y^\alpha \langle S(\cdot,y), \phi_k \rangle = \langle \partial_y^\alpha S(\cdot,y), \phi_k \rangle$ almost everywhere. 
 The second point is shown by observing that
\begin{equation*}
 \begin{split}
  \| y \mapsto \partial_y^\alpha \langle S(\cdot,y), \phi_k) \rangle \|_{L^2(\R^d)}^2 \\
  & = \int_{\R^d} | \langle \partial_y^\alpha  S(\cdot,y), \phi_k) \rangle |^2 dy \\
  & \leq \sum_{k \in \N} \int_{\R^d} | \langle \partial_y^\alpha  S(\cdot,y), \phi_k) \rangle |^2 dy \\
  & = \int_{\R^d} \sum_{k \in \N}  | \langle \partial_y^\alpha  S(\cdot,y), \phi_k) \rangle |^2 dy \\
  & = \int_{\R^d} \| \partial_y^\alpha  S(\cdot,y)\|_{L^2(\R^d)}^2 dy \\
  & = \int_{\R^d} \| \partial_y^\alpha  S(x,\cdot)\|_{L^2(\R^d)}^2 dx.
 \end{split}
\end{equation*}
We used the Tonelli Theorem to switch the sum with the integrals and then the two integrals. Therefore
\begin{equation*}
  \begin{split}
    \| y \mapsto \langle S(\cdot,y), \phi_k) \rangle \|_{H^s(\R^d)}^2 & \leq \sum_{|\alpha| \leq s} \| y \mapsto \partial_y^\alpha \langle S(\cdot,y), \phi_k) \rangle \|_{L^2(\R^d)}^2  \\
    & \leq \int_{\R^d} \sum_{|\alpha| \leq s}  \| \partial_y^\alpha  S(x,\cdot)\|_{L^2(\R^d)}^2 dx \\
    & = \int_{\R^d} \| S(x,\cdot)\|_{H^s(\R^d)}^2 dx \\
    & \leq A_2.
  \end{split}
\end{equation*}

 The third one is straightforward once the following is shown
 \begin{equation*}
  \begin{split}
    \sum_{k \in \N} \| \partial_y^\alpha F(\cdot)[k] \|_{L^2(\R^d)}^2 & = \int_{\R^d} \sum_{k \in \N} | \partial_y^\alpha F(y)[k] |^2 dy \\  
    & = \int_{y \in \R^d} \| \partial_y^\alpha S(\cdot,y) \|^2_{L^2(\R^d)} dy \\
    & = \int_{x \in \R^d} \| \partial_y^\alpha S(x,\cdot) \|^2_{L^2(\R^d)} dx.
  \end{split}
 \end{equation*}
 Note that we switched the sum with the integral then the two integrals using the Tonelli Theorem.
 The last point goes as follows:
 \begin{equation*}
  \begin{split}
     \sum_{k \in \N} w[k] \left\| F(\cdot)[k] \right\|_{L^2(\R^d)}^2 &= \int_{\R^d} \sum_{k \in \N} w[k]  \left| \langle \partial_y^\alpha S(\cdot, y), \phi_k \rangle \right|^2 dy \\
      & = \int_{\R^d} \left\| S(\cdot, y) \right\|_{\mathcal{E}^r(\R^d)}^2 dy.
   \end{split}
  \end{equation*}
Note that we switched the sum and integral using the Tonelli Theorem.
\end{proof}

\subsection{Proof of the main results}

\paragraph{Proof of Theorem \ref{thm:main_result}}

\begin{proof}
By equation \eqref{eq:risk_separation}:
\begin{equation}
 \E \|\hat{H} - H\|_{HS}^2 \leq 2(\epsilon_d(N) + \epsilon_e(n)).
\end{equation}

By Corollary \eqref{cor:epsilond}
\begin{equation}
 \epsilon_d(N) \lesssim N^{-2r/d}.
\end{equation}

Now, let us control $\epsilon_e$. 
\begin{align}
  \epsilon_e(n) &= \E \| \hat{F} - F \|_{\R^N \times L^2(\Omega)}^2 \\
 &= \sum_{k=1}^N \E \| \hat{F}(\cdot)[k] - F(\cdot)[k] \|_{ L^2(\Omega)}^2 \\
  &\stackrel{\eqref{eq:upperbound1}}{\lesssim}  \sum_{k=1}^N \left( \mu (1-\alpha)^{-1}  \| F(\cdot)[k] \|_{\mathcal{G}_k(\R^d)}^2 + n^{-1} \sigma^2 \left[ (1-\alpha) \mu\right]^{-d/2s} (1-\alpha)^{-1} \right) \\
  &= \mu (1-\alpha)^{-1} \sum_{k=1}^N \| F(\cdot)[k] \|_{\mathcal{G}_k(\R^d)}^2 + N n^{-1} \sigma^2 \left[ (1-\alpha) \mu\right]^{-d/2s} (1-\alpha)^{-1}
\end{align}
Further calculations give
\begin{align}
 \sum_{k=1}^N \| F(\cdot)[k] \|_{\mathcal{G}_k(\R^d)}^2 &= (1-\alpha)\sum_{k=1}^N |F(\cdot)[k]|_{BL^s(\R^d)}^2 + \alpha \sum_{k=1}^N w[k] \|F(\cdot)[k]\|_{L^2(\R^d)}^2 \\
 & \leq (1-\alpha) A_1 + \alpha A_2.
\end{align}
where the last inequality is derived using Lemma \eqref{lem:invert}. Hence,
\begin{equation*}
   \epsilon_e(n) \leq \mu (1-\alpha)^{-1} (A_1+A_2) + N n^{-1} \sigma^2 \left[ (1-\alpha) \mu\right]^{-d/2s} (1-\alpha)^{-1}.
\end{equation*}

This upper bound allows to set the value of the regularization parameter $\mu$ by balancing the two terms $ (1-\alpha)^{-1} \mu $ and $N n^{-1} \sigma^2 \left[ (1-\alpha) \mu\right]^{-d/2s} (1-\alpha)^{-1}$:
 \begin{equation}
  (1-\alpha)^{-1} \mu  \propto N n^{-1} \sigma^2 \left[ (1-\alpha) \mu\right]^{-d/2s} (1-\alpha)^{-1}.
 \end{equation}
This yields  
\begin{equation}
\mu \propto \left( N \sigma^2 n^{-1} \right)^{\frac{2s}{2s+d}} (1-\alpha)^{\frac{-d}{2s+d}}. 
\end{equation}
Plugging this value in the upper-bound of $\epsilon_e(n)$ gives
\begin{equation}
    \begin{split}
   \mu (1-\alpha)^{-1} & \propto \left( N \sigma^2 n^{-1} \right)^{\frac{2s}{2s+d}} (1-\alpha)^{-1}(1-\alpha)^{\frac{-d}{2s+d}}.
  \end{split}
\end{equation}
Hence,
\begin{equation}
 \epsilon_e(n) \lesssim (N \sigma^2 n^{-1})^{\frac{2s}{2s+d}} (1-\alpha)^{-\frac{2s+2d}{2s+d}}.
\end{equation}
\end{proof}


\paragraph{Proof of Corollary \ref{thm:main_result2}}

\begin{proof}
 To obtain this bound we use Theorem \ref{thm:main_result} and we balance the two terms so that:
 \begin{equation}
  N^{-2r/d} \propto (N \sigma^2 n^{-1})^{\frac{2s}{2s+d}} (1-\alpha)^{-\frac{2s+2d}{2s+d}}
 \end{equation}
This gives the choice $N \propto (\sigma^{-2} n)^{\frac{2sd}{4rs + 2rd + 2sd}} (1-\alpha)^{\frac{(2s+2d)d}{4rs + 2rd + 2sd}} $. Replacing $N$ by this value in bound \eqref{eq:main_result} gives
\begin{equation}
  \begin{split}
    N^{-2r/d} & \propto (\sigma^{2} n^{-1})^{\frac{4rs}{4rs + 2rd + 2sd}} (1-\alpha)^{-\frac{4rd + 4rs}{4rs + 2rd + 2sd}} , \\
    &= \left(\sigma^{2} n^{-1} (1-\alpha)^{-\left(d/s+1\right)} \right)^{\frac{2q}{2q+d}}.
 \end{split}
\end{equation}
\end{proof}

\paragraph{Proof of Theorem \ref{thm:minimax}}

\begin{proof}
We first need to define an appropriate wavelet basis to characterize the fact that a function belongs to the Sobolev ball (for some constant $A > 0$)
$$
H^s(\Omega,A) = \left\{ u\in L^2(\Omega), \; \|u\|_{H^s(\Omega)}^{2}  \leq A\right\},
$$
through its wavelet coefficients.
 The scaling and wavelet functions at scale $j$ (that is at resolution level $2^{j}$) will be denoted by $\phi_{\lambda}$ and $\psi_{\lambda}$, respectively, where the index $\lambda$ summarizes both the usual scale and space parameters $j$ and $k$. In other words, for $d=1$, we set $\lambda = (j,k)$ and denote $\phi_{j,k} (\cdot)= 2^{j/2} \phi(2^{j}\cdot -k)$ and $\psi_{j,k}(\cdot) = 2^{j/2} \psi(2^{j}\cdot -k)$. For $d \geq 2$, the notation $\psi_{\lambda}$ stands for the adaptation of scaling and wavelet functions to $\Omega = [0,1]^{d}$ (see  \cite{cohen2003numerical}, Chapter 2). The notation $|\lambda| = j$ will be used to denote a wavelet at scale $j$,  where $j_{0}$ denotes the coarse level of approximation. In order to simplify the notation, as it is commonly used, we take $j_{0} = 0$, and we write $(\psi_{\lambda})_{|\lambda| = -1}$ for $(\phi_{\lambda})_{|\lambda| = 0}$. Finally, $|\lambda| < j_{1}$ denotes all wavelets at scales $j$, with $-1 \leq j < j_{1}$, and we use the notation $\tilde{\psi}_{\lambda}$ to denote the dual wavelet basis of $\psi_{\lambda}$. Now, assume that a function $u \in L^{2}(\Omega)$ admits the wavelet decomposition
\begin{equation}
u(y) =  \sum_{j = -1}^{+\infty} \sum_{|\lambda| = j} c[\lambda] \psi_{\lambda}(y) \label{eq:decompu}
\end{equation}
where the $c[\lambda]$'s are real coefficients satisfying $c[\lambda] = \langle u, \tilde{\psi}_{\lambda} \rangle_{L^{2}(\Omega)}$. It is well known that wavelet coefficients may be used to characterize the smoothness of functions. For instance, by Theorem 3.10.5 in \cite{cohen2003numerical} (on the equivalence of norms between Besov and sequence of wavelet coefficients spaces) and  using the fact that the Besov space $B^{s}_{2,2}(\Omega)$ is equal to the Sobolev space $H^s(\Omega)$ (see e.g.\ Remark 3.2.4 in \cite{cohen2003numerical}), it follows that, under appropriate assumptions on the scaling function $\phi$ and its dual version  (see e.g.\ those of Theorem 3.10.5 in \cite{cohen2003numerical}),  there exist two  constants $C_1(s,\Omega) > 0$ and $C_2(s,\Omega) > 0$ (depending only on $s$ and $\Omega$) such that, for any $u$ admitting the decomposition \eqref{eq:decompu},
\begin{equation}
C_1(s,\Omega) \sum_{j = -1}^{+\infty} \sum_{|\lambda| = j} 2^{2js} |c[\lambda]|^{2} \leq \|u\|_{H^s(\Omega)}^{2} \leq C_2(s,\Omega) \sum_{j = -1}^{+\infty} \sum_{|\lambda| = j} 2^{2js} |c[\lambda]|^{2}. \label{eq:condSobolevBall}
\end{equation}
Throughout the proof, it is assumed that the bi-orthogonal wavelet basis is chosen such that the wavelet characterization of Sobolev norms \eqref{eq:condSobolevBall} is satisfied.
In particular, we assume that $\psi$ possesses $s+1$ vanishing moments.

The arguments to prove the lower bound \eqref{eq:lowbound} are based on the standard Assouad's cube technique (see, e.g.\ \cite{MR2724359}, Chapter 2, Section 2.7.2). 
Assuming that the wavelets $\psi_{\lambda}$ are extended outside $\Omega$ using the convention that $\psi_{\lambda}(y) = 0$ for $y \notin \Omega$, we consider the following SVIR test functions 
$$
S_{v}(x,y) = \mu_{k_{1},j_{1}} \sum_{k=1}^{k_{1}}   \sum_{|\lambda| < j_{1}}  v[k,\lambda] \phi_{k}(x) \psi_{\lambda}(y) , \quad \forall (x,y) \in \R^d \times \R^d,
$$
where $v = \left( v[k,\lambda] \right)_{k \leq k_{1}, |\lambda| < j_{1} }  \in \mathcal{V} := \{1,-1\}^{k_{1} 2^{j_{1} d}}$,  and $\mu_{k_{1},j_{1}}$ is a positive sequence of reals satisfying the condition
\begin{equation}
 \mu_{k_{1},j_{1}}  =  c k_{1}^{-1/2} 2^{-j_{1}d/2} \min \left( k_{1}^{-r / d}, 2^{-j_{1} s} \right),\label{eq:condmu}
\end{equation}
for some constant $c > 0$  not depending on $k_{1}$ and $j_{1}$.  
Note that, for all $x \in  \R^d$, $S_{v}(x, \cdot)$ is compactly supported in $\Omega$.

Let us first discuss the choice of the constant $c$ in \eqref{eq:condmu}. It has to be chosen such that $S_{v}$ is the SVIR of an operator in $\mathcal{E}^{r,s}(A_1,A_2)$.

First, for any $v \in \mathcal{V}$ one has that:
\begin{eqnarray*}
\mu_{k_1,j_1}^2 \int_{\R^d} \sum_{j = -1}^{j_{1}-1} \sum_{|\lambda| = j} 2^{2js} \left( \sum_{k=1}^{k_1} v[k,\lambda] \phi_k(x) \right)^2 dx
 & = &   \mu_{k_1,j_1}^2 \sum_{j = -1}^{j_{1}-1} \sum_{|\lambda| = j} 2^{2js}  \int_{\R^d} \sum_{k,l=1}^{k_1} v[k,\lambda] v[l,\lambda] \phi_k(x) \phi_l(x) dx \\
 & = &   \mu_{k_1,j_1}^2  \sum_{j = -1}^{j_{1}-1} 2^{2js}  \sum_{|\lambda| = j}  \sum_{k=1}^{k_1} v[k,\lambda]^2 \| \phi_k \|_{L^2(\R^d)}^2 \\ 
 & = &  \mu_{k_1,j_1}^2 \sum_{j=-1}^{j_1-1} 2^{j(2s+d)} k_1  \leq    \mu_{k_1,j_1}^2 k_1 2^{j_1(2s+d)}.
\end{eqnarray*}
where the above equalities use the orthonormality of the basis $(\phi_k)$, the definition $v[k,\lambda] = \pm 1$ and the fact that the number of wavelets at scale $j$ is $2^{jd}$. Now, using that $\int_{\R^d} \| S_{v}(x,\cdot)\|_{H^s(\R^d)}^2 dx = \int_{\R^d} \| S_{v}(x,\cdot)\|_{H^s(\Omega)}^2 dx$, the wavelet characterization of Sobolev norms \eqref{eq:condSobolevBall} and by the condition \eqref{eq:condmu} on $\mu_{k_{1},j_{1}}$, it follows that if $c^{2} \leq C_2^{-1}(s,\Omega) A_{1}$ then
$$
\int_{x \in \R^d} \|S_{v}(x,\cdot)\|_{H^s(\R^d)}^{2} dx \leq A_{1},
$$
for any $v \in \mathcal{V}$. 

We now proceed to the other inequality. A key element is again that $S_{v}(x, \cdot)$ is compactly supported in $\Omega$ for any  $x \in \R^{d}$.
For any $v \in \mathcal{V}$:
\begin{equation}
  \begin{split}
      \int_{y \in \R^d} \|S_{v}(\cdot,y) \|_{\mathcal{E}^r(\R^d)}^2 dy & = \int_{y \in \Omega} \|S_{v}(\cdot,y) \|_{\mathcal{E}^r(\R^d)}^2 dy  \\
      & = \int_{y \in \Omega} \sum_{k\in \N} w[k] \left|\langle S_{v}(\cdot,y) , \phi_k \rangle \right|^2 dy \\
      & = \int_{y \in \Omega} \sum_{k=1}^{k_1} w[k] \mu_{k_1,j_1}^2 \left( \sum_{j = -1}^{j_{1}-1} \sum_{|\lambda| = j}  v[k,\lambda] \psi_{\lambda}(y) \right)^2 dy \\
      & \leq c_1 \mu_{k_1,j_1}^2  \int_{y \in \Omega} \sum_{k=1}^{k_1}(1+k^2)^{r/d} \left( \sum_{j = -1}^{j_{1}-1} \sum_{|\lambda| = j}  v[k,\lambda] \psi_{\lambda}(y) \right)^2 dy, \label{eq:bound1}
        \end{split}
\end{equation}
using the assumption that  $w[k] \leq c_{1}  (1+k^2)^{r/d}$. Let us define the set
 $$
 I_{j}(y) = \{ \lambda \; : \;  |\lambda| = j \mbox{ and }  \psi_{\lambda}(y) \neq 0 \}.
 $$
 Due to the compact support of $\psi_{\lambda}$ the cardinality of $I_{j}(y)$ is bounded by a  constant $D_s > 0$ that is independent of $j$ and $y$. Thus using the relation $ \|\psi_{\lambda} \|_{\infty} \leq  C_{\infty} 2^{j d /2}$ (for some constant $C_{\infty}$ > 0) for any $\lambda$ at scale $j$, we obtain from inequality \eqref{eq:bound1} and the fact that $\left(  \int_{y \in \Omega} dy \right) = 1$ (since $\Omega = [0,1]^{d}$ in this proof), 
 \begin{eqnarray*}
\int_{y \in \R^d} \|S_{v}(\cdot,y) \|_{\mathcal{E}^r(\R^d)}^2 dy & \leq &  c_{1} \mu_{k_{1},j_{1}}^{2}  \int_{y \in \Omega}  \sum_{k=1}^{k_{1}}  (1+k^2)^{r/d} \left(  \sum_{j = -1}^{j_{1}-1} \sum_{\lambda \in I_{j}(y)}  |\psi_{\lambda}(y)|  \right)^{2} dy \\
 & \leq & c_{1} D_s^{2} C_{\infty}^{2}  \mu_{k_{1},j_{1}}^{2} \sum_{k=1}^{k_{1}}  (1+k^2)^{r/d} \left(  \sum_{j = -1}^{j_{1}-1}  2^{j d /2} \right)^{2}, \\
 & \leq & c_{1} D_s^{2} C_{\infty}^{2} \mu_{k_{1},j_{1}}^{2} k_{1}^{2r/d+1}   2^{j_{1} d}.
 \end{eqnarray*}
Hence, by the condition \eqref{eq:condmu} on $\mu_{k_{1},j_{1}}$  it follows that if 
$c^{2} \leq A_{2} c_{1}^{-1} D_s^{-2} C_{\infty}^{-2} $, then
$$
 \int_{y \in \R^d} \|S_{v}(\cdot,y) \|_{\mathcal{E}^r(\R^d)}^2 dy \leq A_{2},
$$
for any $v \in \mathcal{V}$. Thus we have shown that if the constant $c$ in  \eqref{eq:condmu}  is chosen sufficiently small, then the operator  $H_{v}$ with SVIR  function $S_{v}$ belongs to the ball $\mathcal{E}^{r,s}(A_{1},A_{2})$ for any $v \in \mathcal{V}$. In the rest of the proof, the constant $c$ is assumed to be chosen in such a manner.

In what follows, we use the notation $\E_{H_{v}}$ to denote expectation with respect to the distribution $\P_{H_{v}}$ of the random process $F^\epsilon = (F_{1}^\epsilon,\ldots,F_{n}^\epsilon)$ obtained from model \eqref{eq:datamodel}  under the hypothesis that $S = S_{v}$ where $S_{v}$ is the SVIR function of the operator $H_{v}$.

The minimax risk
$$
\mathcal{R}_{\sigma^2,n}:= \inf_{\hat H } \sup_{H \in \mathcal{E}^{r,s}(A_{1},A_{2})}\E \left\| \hat H - H \right\|_{HS}^2
$$
can be bounded from below as follows  
$$
\mathcal{R}_{\sigma^2,n}  \geq  \inf_{\hat H } \sup_{v \in  \mathcal{V}} \E_{H_{v}} \left\| \hat H - H_{v} \right\|_{HS}^2.
$$
Since $\left\| \hat H - H_{v} \right\|_{HS}^2 = \left\| \hat{S} - S_v \right\|_{L^2(\R^d \times \Omega)}^2$ it follows from orthonormality of the bases $(\phi_{k})_{k \geq 1}$ and the Riesz stability property for bi-orthogonal wavelet bases (see e.g.\ inequality (7.156) in \cite{Mallat-Book}), that there exists a constant $c_{\psi} > 0$ such that
$$
 \left\| \hat H - H_v \right\|_{HS}^2 \geq c_{\psi} \sum_{k=1}^{k_{1}} \sum_{|\lambda| < j_{1}} |\hat{\alpha}[k,\lambda] - \mu_{k_{1},j_{1}}   v[k,\lambda]|^2 \mbox{ where } \hat{\alpha}[k,\lambda] = \int_{\R^d \times \Omega} \hat{S}(x,y)  \phi_{k}(x) \psi_{\lambda}(y) dx dy.
$$
Therefore, the minimax risk satisfies the following inequality
$$
\mathcal{R}_{\sigma^2,n}  \geq  \inf_{\hat H } \sup_{v \in  \mathcal{V}} c_{\psi}  \sum_{k=1}^{k_{1}} \sum_{|\lambda| < j_{1}} \E_{H_{v}} \left|\hat{\alpha}[k,\lambda] - \mu_{k_{1},j_{1}}   v[k,\lambda] \right|^2.
$$
Then, define
$$
 \hat{v}[k,\lambda] := \argmin_{v \in \{-1,1\}} \left| \hat{\alpha}[k,\lambda]  - \mu_{k_{1},j_{1}}  v \right|,
$$
and  remark that the triangular inequality and the definition of $ \hat{v}[k,\lambda]$ imply that
$$
\mu_{k_{1},j_{1}} \left| \hat{v}[k,\lambda] -    v[k,\lambda] \right| \leq 2 \left| \hat{\alpha}[k,\lambda]  -  \mu_{k_{1},j_{1}}   v[k,\lambda] \right|,
$$
which yields
\begin{eqnarray*}
\mathcal{R}_{\sigma^2,n} & \geq &  \inf_{\hat{H}  } \sup_{v \in  \mathcal{V}}  \frac{c_{\psi} \mu_{k_{1},j_{1}}^2}{4}   \sum_{k=1}^{k_{1}} \sum_{|\lambda| < j_{1}} \E_{H_{v}} \left|  \hat{v}[k,\lambda]  -      v[k,\lambda] \right|^2 \\
& \geq &  \inf_{\hat{H}  }   \frac{ c_{\psi} \mu_{k_{1},j_{1}}^2}{4}  \frac{1}{\# \mathcal{V}} \sum_{v \in  \mathcal{V}} \sum_{k=1}^{k_{1}} \sum_{|\lambda| < j_{1}}   \E_{H_{v}} \left|  \hat{v}[k,\lambda]  -      v[k,\lambda] \right|^2,
\end{eqnarray*}
where  $\#  \mathcal{V}$ denotes the cardinality of the finite set $ \mathcal{V}$. 

For a given pair $[k,\lambda]$ of indices and any $v \in \mathcal{V}$, we define the vector $v^{(k,\lambda)} \in \mathcal{V}$ having all its components equal to $v$ except the $[k,\lambda]$-th element. Moreover, to simplify the notation,  we let $\sum_{k,\lambda}$ denote the summation $ \sum_{k=1}^{k_{1}} \sum_{|\lambda| < j_{1}}$. Then
\begin{eqnarray*}
\mathcal{R}_{\sigma^2,n} & \geq &  \inf_{\hat{H}  }    \frac{ c_{\psi} \mu_{k_{1},j_{1}}^2}{4}  \frac{1}{\# \mathcal{V}} \sum_{k,\lambda}  \sum_{v \in  \mathcal{V} \; : \;  v[k,\lambda]  = 1} \left(  \E_{H_{v}} \left|  \hat{v}[k,\lambda]  -       v[k,\lambda] \right|^2 +  \E_{H_{v^{(k,\lambda)}}} \left|  \hat{v}[k,\lambda]  -      v^{(k,\lambda)}[k,\lambda] \right|^2 \right)  \\
 & \geq &  \inf_{\hat{H}  }    \frac{ c_{\psi} \mu_{k_{1},j_{1}}^2}{4}  \frac{1}{\# \mathcal{V}} \sum_{k,\lambda}  \sum_{v \in  \mathcal{V} \; : \;  v[k,\lambda]  = 1}  \E_{H_{v}} \left(  \left|  \hat{v}[k,\lambda]  -       v[k,\lambda] \right|^2 +   \left|  \hat{v}[k,\lambda]  -       v^{(k,\lambda)}[k,\lambda] \right|^2 \frac{d \P_{H_{v^{(k,\lambda)}}}}{d \P_{H_{v}}} (F^\epsilon)  \right).
\end{eqnarray*}
where  $\frac{d \P_{H_{v^{(k,\lambda)}}}}{d \P_{H_{v}}}(F^\epsilon)$ is the log-likelihood ratio between the hypothesis $H_{v^{(k,\lambda)}} : S = S_{v^{(k,\lambda)}}$ and the hypothesis $H_{v} : S = S_{v}$ in model \eqref{eq:datamodel}.

Since $ v^{(k,\lambda)}[k,\lambda] = - v[k,\lambda] $ and $\hat{v}[k,\lambda]  \in \{-1,1\}$, one has that, for any $0 < \delta < 1$,
\begin{eqnarray}
\mathcal{R}_{\sigma^2,n} & \geq &  4 c_{\psi} \mu_{k_{1},j_{1}}^2   \frac{1}{\# \mathcal{V}} \sum_{k,\lambda}  \sum_{v \in  \mathcal{V} \; : \;  v[k,\lambda]  = 1}  \;\E_{H_{v}} \left(  \min\left(1,   \frac{d \P_{H_{v^{(k,\lambda)}}}}{d \P_{H_{v}}} (F^\epsilon) \right)   \right) \nonumber \\
 & \geq &  4 \delta c_{\psi} \mu_{k_{1},j_{1}}^2   \frac{1}{\# \mathcal{V}} \sum_{k,\lambda}  \sum_{v \in  \mathcal{V} \; : \;  v[k,\lambda]  = 1}\;  \P_{H_{v}} \left(   \frac{d \P_{H_{v^{(k,\lambda)}}}}{d \P_{H_{v}}} (F^\epsilon) > \delta  \right), \label{eq:lowerdelta}
\end{eqnarray}
by Markov's inequality. Thanks to the Girsanov's formula (see e.g.\ Lemma A.5 in \cite{MR2724359}), one has that, under the hypothesis that $S = S_{v}$  in model \eqref{eq:datamodel}:
$$
\log \left(  \frac{d \P_{H_{v^{(k,\lambda)}}}}{d \P_{H_{v}}} (F^\epsilon) \right) = \sum_{i = 1}^{n} \sum_{\ell=1}^{+ \infty} \left( \sigma^{-1}  \langle S_{v^{(k,\lambda)}}(\cdot, y_i) - S_{v}(\cdot, y_i), \phi_{\ell} \rangle \eta_{i,\ell} - \frac{\sigma^{-2}}{2} \left| \langle S_{v^{(k,\lambda)}}(\cdot, y_i) - S_{v}(\cdot, y_i), \phi_l \rangle \right|^2 \right)
$$
where  the $\eta_{i,\ell}$'s are iid standard Gaussian variables. By definition of $v^{(k,\lambda)}$ and for $v[k,\lambda] = 1$ one has that, for each $1 \leq i \leq n$ and $1 \leq k \leq k_{1}$,
$$
\langle  S_{v^{(k,\lambda)}}(\cdot, y_i) - S_{v}(\cdot, y_i) , \phi_{\ell} \rangle =
\left\{
\begin{array}{cl}
- 2 \mu_{k_{1},j_{1}}  \psi_{\lambda}(y_i) & \mbox{ if } \ell = k, \\
0 & \mbox{ otherwise. } 
\end{array}
\right.  
$$
Therefore, the random variable 
$
Z_{k,\lambda} : = \log \left(  \frac{d \P_{H_{v^{(k,\lambda)}}}}{d \P_{H_{v}}} (F^\epsilon) \right)
$
is Gaussian with mean $\theta_{\lambda}$ and variance $\gamma^2_{\lambda}$ satisfying
$$
\theta_{\lambda} = -  2  \sigma^{-2} \mu_{k_{1},j_{1}}^2 \sum_{i=1}^{n} \psi_{\lambda}^{2}(y_{i}) \mbox{ and } \gamma^{2}_{\lambda} = 4  \sigma^{-2} \mu_{k_{1},j_{1}}^2 \sum_{i=1}^{n} \psi_{\lambda}^{2}(y_{i})  = -2 \theta_{\lambda},
$$
under the hypothesis  that $S = S_{v}$  in model \eqref{eq:datamodel}. 
The negativity of $\theta_{\lambda}$ implies that
$$
\P_{H_{v}} \left( Z_{k,\lambda} \geq 3 \theta_{\lambda} \right) = \P_{H_{v}} \left( \frac{Z_{k,\lambda}-\theta_{\lambda}}{\sqrt{2 |\theta_{\lambda}|}} \geq -\sqrt{2 |\theta_{\lambda}|} \right) \geq \frac{1}{2},
$$
by symmetry of the standard Gaussian distribution. Hence, inserting the above inequality into \eqref{eq:lowerdelta} with $\delta=\exp(3\theta_\lambda)$, it implies that
\begin{align}
\mathcal{R}_{\sigma^2,n} &  \geq c_{\psi} \exp(3 \theta_{\lambda})  \mu_{k_{1},j_{1}}^2 \;  k_{1}  C_\psi 2^{d j_{1}} \\
& = c_{\psi} \exp(3 \theta_{\lambda}) c^2 C_\psi \min\left(k_1^{-2r/d}, 2^{-2j_1s} \right). \label{eq:boundR}
\end{align}

By setting $k_{1} = k_{1}^{(\sigma^2,n)}$ and $j_{1} = j_{1}^{(\sigma^2,n)}$ with
\begin{equation}
\label{eq:k1j1}
k_{1}^{(\sigma^2,n)} = \lfloor \left(\sigma^2 n^{-1} \right)^{-\frac{q }{ (2q + d)r/d}}  \rfloor
\quad \mbox{and} \quad
2^{j_{1}^{(\sigma^2,n)}} = \lfloor \left(\sigma^2 n^{-1} \right)^{-\frac{q}{ (2q + d)s}} \rfloor,
\end{equation}
we get 
\begin{equation} \label{eq:last}
\mathcal{R}_{\sigma^2,n} \geq  c^2 c_\psi C_\psi \exp(3 \theta_{\lambda}) \left(\sigma^2 n^{-1} \right)^{\frac{2q}{ 2q + d}}.
\end{equation}
It now remains to show that the constant  $\theta_{\lambda}$ is bounded from below, independently of $\sigma$ and $n$. 

The idea is to observed that  $\frac{1}{n} \sum_{i=1}^{n} \psi_{\lambda}^{2}(y_{i})$ behaves like a Riemann integral of $\psi_\lambda$ and should therefore be bounded by a constant since $\|\psi_\lambda\|_2=1$.
This statement can be proved using the following reasoning.
Since vector $Y=(y_{1},\ldots,y_{n})$ of PSFs locations  satisfies the quasi-uniformity condition $h_{Y,\Omega} \leq B q_{Y,\Omega}$, we get from Proposition \ref{prop:equiv_fill_distance} that the separation distance $q_{Y,\Omega}$ satisfies $q_{Y,\Omega}^d \geq  B_1 n^{-1}$ for some constant $B_1$. Now, the support of wavelet $\psi_\lambda$ is contained in a hypercube of volume proportional to $2^{-d|\lambda|}$. Hence, the number of locations $y_i$ in $\supp(\psi_\lambda)$ is bounded above by $2^{-d|\lambda|} /q_{Y,\Omega}^d \leq B_2 n 2^{-d|\lambda|}$ for some constant $B_2$. To conclude, notice that $\|\psi_\lambda\|_\infty = 2^{d |\lambda|/2} \| \psi \|_\infty$, hence:
\begin{align*}
 \frac{1}{n} \sum_{i=1}^n \psi_\lambda(y_i)^2 & =\frac{1}{n} \sum_{y_i \in \supp(\psi_\lambda)} \psi_\lambda(y_i)^2 \\
 &\leq \frac{1}{n} B_2 n 2^{-d|\lambda|} \|\psi_\lambda\|_\infty^2 \\
 &\leq B_2 \|\psi\|_\infty^2=:B_3.
\end{align*}
This implies that
$$
\theta_{\lambda} \geq - 2 B_3 c_{1}, \mbox{ for all } \lambda < j_{1}^{(\sigma^2,n)}.
$$
Hence, inserting the above inequality into \eqref{eq:last},
we finally obtain that there exists a constant $c_{0} > 0$, that does not depend on $\frac{\sigma^{2}}{n}$, such that
$$
\mathcal{R}_{\sigma^2,n} \geq c_{0} \left( \sigma^{2} n^{-1} \right)^{\frac{2 q}{ 2q + d}},
$$
completing the proof of the theorem.
\end{proof}

\bibliographystyle{abbrv}
\bibliography{PSF_RKHS}

\end{document}